%% file: reciprocal_Part1_051810.tex
\documentclass[twoside,12pt,reqno]{amsart}
\usepackage{a4wide}
\usepackage{latexsym}
\usepackage{amsfonts}
\usepackage{amssymb}
\usepackage{amsthm}
\usepackage{amsmath}              
\usepackage[noadjust]{cite}       
\usepackage{xypic}
\usepackage{wasysym}              
\input{Tmacros}  
\input{diagxy}   

\newcommand{\ed}{\end{document}}
\def\ve{\varepsilon}
\def\vp{\varphi}

\newcommand{\keyw}[1]{{\bf #1}}
\newcommand{\ta}[2]{#1_#2\tilde{\phantom{.}}}
\newcommand{\cb}[1]{\mathcal{#1}}
\newcommand{\matrepr}[1]{\left[#1 \right]}

\newcommand{\iu}{{\underline{i}}}
\newcommand{\ju}{{\underline{j}}}
\newcommand{\ku}{{\underline{k}}}
\newcommand{\lu}{{\underline{l}}}
\newcommand{\Hom}{\mathrm{Hom}}
\newcommand{\End}{\mathrm{End}}

\newcommand{\cml}{{\mathrm{cmul}}}
\newcommand{\spl}{\mathrm{split}}

\newcommand{\tve}{t_\ve}
\newcommand{\tp}{\ta{T}{\ve}}

\newcommand{\E}[2]{\mathrm{E}^{#1}_{#2}}
\newcommand{\F}[2]{\mathrm{F}^{#1}_{#2}}

\newcommand{\lang}{\mathopen{<}}
\newcommand{\rang}{\mathclose{>}}

\renewcommand{\qedsymbol}{\APLbox}

\DeclareMathOperator{\hotimes}{\Hat{\otimes}}

\theoremstyle{plain}
\newtheorem{theorem}{Theorem}
\newtheorem{corollary}{Corollary}
\newtheorem{lemma}{Lemma}
\newtheorem{proposition}{Proposition}
\theoremstyle{definition}
\newtheorem{definition}{Definition}
\newtheorem{example}{Example}
\newtheorem*{remark}{Remark}

\begin{document}
\title[On the Transposition Anti-Involution in Real Clifford Algebras I: ...]
{On the Transposition Anti-Involution in Real\\ Clifford Algebras I:
The Transposition Map}
\author[Rafa\l \ Ab\l amowicz]{Rafa\l \ Ab\l amowicz}

\address{%
Department of Mathematics, Box 5054\\
Tennessee Technological University\\
Cookeville, TN 38505, USA}
\email{rablamowicz@tntech.edu}

\author[Bertfried Fauser]{Bertfried Fauser}
\address{%
School of Computer Science\\
The University of Birmingham\\
Edgbaston-Birmingham, W. Midlands, B15 2TT\\
United Kingdom}
\email{B.Fauser@cs.bham.ac.uk}



\date{\today}


\begin{abstract} A particular orthogonal map on a finite dimensional real quadratic vector space $(V,Q)$ with a non-degenerate quadratic form $Q$ of any signature $(p,q)$ is considered. It can be 
viewed as a correlation of the vector space that leads to a dual Clifford algebra $\cl(V^\ast,Q)$ of linear functionals (multiforms) acting on the universal Clifford algebra $\cl(V,Q).$  The map 
results in a unique involutive automorphism and a unique involutive anti-automorphism of $\cl(V,Q).$ The anti-involution reduces to reversion (resp. conjugation) for any Euclidean (resp. anti-Euclidean)
signature. When applied to a general element of the algebra, it results in transposition of the element matrix in the left regular representation of $\cl(V,Q)$. We give also an example for real spinor spaces. The general setting for spinor representations will be treated in part II of this work~\cite{part2}.
\end{abstract}
\maketitle

\section{Introduction}\label{introduction}
\medskip

In this paper we study a particular involution and a particular anti-involution of a universal real Clifford algebra $\cl(V,Q)$ of a non-degenerate quadratic real vector space of dimension $n.$ Except for the non-degeneracy of the quadratic form $Q$, we do not assume any particular signature $(p,q)$ of~$Q$ and we state our results for all signatures. Thus, we denote the universal $2^n$-dimensional Clifford algebra of~$(V,Q)$ by $\cl_n$ which should not be confused with $\cl_{n,0}$ or $\cl_{0,n}$, that is, respectively, with the Clifford algebra of the Euclidean or the anti-Euclidean quadratic 
space~$V.$ For the structure theory of Clifford algebras we refer to \cite{lam, helmmicali, lounesto, porteous} and, in particular, we adopt Chevalley's definition of the Clifford algebra $\cl(V,Q)$ as 
a subalgebra of the algebra of endomorphisms of the exterior algebra $\bigwedge V$~\cite{lounesto}. For the study of correlations and involutions we refer to \cite{porteous}.

We adopt (with a small modification) Porteous' definition \cite{porteous} of an $\bL^\alpha$-Clifford algebra where $\bL^\alpha$ denotes a \textit{superfield}. In general, a \textit{superfield} 
$\bL^\alpha$, with a fixed field~$k$, consists of a commutative algebra $\bL$ with a unit element $1$ over a commutative field $k$ and an involution $\alpha$ of $\bL,$ whose set of fixed points is the set of scalar multiples of $1$, identified as usual with~$k.$\footnote{Of course, as a special case of the superfield $\bL^\alpha$, we take in this paper $\bL^\alpha=\BR^\alpha$ with $\alpha$ being 
the identity map on $\BR.$ However, in general, the algebra $\bL$ itself need not be a field.} 

Let $(V,Q)$ be a finite-dimensional quadratic space over a commutative field~$k$ and, let~$A$ be an associative $\bL$-algebra with unity $1_A$ where $\bL$ is identified with a subalgebra of~$A.$ Then, 
$A$ is said to be an $\bL^\alpha$-\textit{Clifford algebra for $(V,Q)$} if it contains~$V$ as a $k$-linear subspace in such a way that, for all $\bx \in V$, $\bx^2 = Q(\bx)\cdot 1_A$, and provided also that $A$ is generated as a ring by $\bL$ and $V.$~\footnote{In Porteous' definition, $V$ is a quadratic space over some not necessarily real field $k$ and he uses the negative sign in 
$\bx^2 = -Q(\bx) \cdot 1_A$.} Let $\cl(V,Q)$ be the universal (real) Clifford algebra and let $A$ be the universal $\bL^\alpha$-Clifford algebra where $\bL^\alpha$ is a superfield with a fixed field $k.$ Then we have  $\cl(V,Q) \cong A \otimes_k \bL$~\cite[Prop. 15.31]{porteous}.  

The following theorem from Porteous~\cite[Thm. 15.32]{porteous} modified later to the  case when $k=\BR$ and $\bL^\alpha=\BR^\alpha$ with $\alpha$ being the identity map on $\BR$, is fundamental to 
this work\footnote{Here, the term $k$-\textit{orthogonal space} $V$ is  a synonym for a $k$-quadratic space $V$, or, the pair $(V,Q)$ in our notation.}: 
\begin{theorem}[Porteous]
Let $A$ be an universal $\bL^\alpha$-Clifford algebra for a finite-dimensional $k$-orthogonal space $V$, $\bL^\alpha$ being a superfield with involution $\alpha$ and fixed commutative field  $k$. Then for any orthogonal automorphism $t : V \rightarrow V$, there is a unique $\bL$-algebra automorphism $t_A:A \rightarrow A$, sending any $\lambda \in \bL$ to $\lambda$, and a unique $k$-algebra anti-automorphism $\ta{t}{A}: A \rightarrow A$ sending any $\lambda$ to $\lambda^\alpha$, such that the diagrams 
\[
\bfig
\square|arra|/>` >->` >->`>/[V`V`A`A;t`\iota`\iota`t_A]
\efig
\qquad \mbox{and} \qquad
\bfig
\square|arra|/>` >->` >->`>/[V`V`A`A;t`\iota`\iota`\ta{t}{A}]
\efig
\]
commute. Moreover, $(1_V)_A = 1_A$ and, for any $t,u \in O(V)$,
\begin{align}
(ut)_A &= u_At_A=\ta{u}{A} \ta{t}{A}.
\label{eq:com1}
\end{align}       
\end{theorem}

If $t$ is an orthogonal \textit{involution} of $V$, then $t_A$ is the \textit{involution} of $A$ and $\ta{t}{A}$ is the \textit{anti-involution} of $A.$ In the above theorem, the map 
$\xy \morphism/ >->/<500,0>[V`A;\iota]\endxy$  is the injection of the vector space $V$ into $A$, and $\lambda^\alpha = \alpha(\lambda), \, \lambda \in \bL.$ Two special cases are well known: When 
$t = -1_V$ is the negative identity linear orthogonal involution on $V$, then the involution of $\cl(V,Q)$ induced by $t$ is the \textit{main involution} or the \textit{grade involution}, and it is denoted by $u \mapsto \hat{u}$. The anti-involution of $\cl(V,Q)$ induced by $t = 1_V$ is called the \textit{reversion} and it is denoted by $u \mapsto \tilde{u}.$\footnote{The reversion symbol 
$\tilde{u}, \, u \in \cl(V,Q)$, should not be confused with Porteous' notation for the anti-automorphism $\ta{t}{A}: A \rightarrow A$.} Finally, the  anti-involution of $\cl(V,Q)$ induced by 
$t = -1_V$ is called the \textit{(Clifford) conjugation} and it is denoted by $u \mapsto \bar{u}$~\cite{lounesto,porteous}. 

In Section~2, we introduce definitions and notation. 

In Section~3, we define a particular orthogonal involution $t_\varepsilon$ of $V$, where $\ve$ denotes the signature of $Q$ over $\mathbb{R}$, and compute the unique algebra involution and the unique algebra anti-involution $\tp$ of $\cl(V,Q)$ which $t_\ve$ induces. We show that this anti-involution gives the ordinary transposition of matrices in the left-regular representation of 
$\cl(V,Q)$ \emph{for any signature} of~$Q$. 

In Section~4, we introduce a universal Clifford algebra $\cl(V^\ast,Q)$ of the dual space $V^\ast$, endowed with the same quadratic form $Q$, of linear forms on $\cl(V,Q).$ This Clifford algebra will be generated by the \textit{inverse} (or \textit{reciprocal}) orthonormal basis of $V^\ast$. It is important not to confuse this dual Clifford algebra with the linear dual Clifford algebra $\cl^*(V,Q)$, which we do not consider throughout this work. We show that this unique anti-involution $\tp$ of $\cl(V,Q)$ reduces to reversion (resp. conjugation) for Euclidean (resp. anti Euclidean) signatures. Furthermore, it is responsible for transposition of matrices in the left-regular representation of $\cl(V,Q).$ 

Section~5 includes examples which illustrate results from the previous sections. We give examples for left regular representations and for real spinor spaces. The general spinor case is deferred to part II of this work~\cite{part2}.
 
In Section~6, we put our results into a bigger picture and explain how this work may be generalized and elaborated. 

After having studied in the present article the left regular representation and its transposition anti-automorphism, one finds a far more interesting problem how the involution $\tp$ acts on irreducible spinor representations. Such representations develop further a particularity that the spinor spaces may need a field extension to complex numbers, quaternions or even double (skew) fields. We will investigate this situation in the second paper \cite{part2}.

\section{Basic definitions and notation}\label{definitions}
\medskip

Let $V$ be an $n$-dimensional real vector space endowed with a non-degenerate quadratic form $Q$ such that
\begin{align}
Q(\bx) &= \ve_1x_1^2 + \ve_2x_2^2 + \cdots + \ve_nx_n^2,
\label{eq:Q}     
\end{align}
where $\ve_i = \pm 1$ and $\bx = x_1\be_1 + \cdots + x_n\be_n \in V$ for a particular orthonormal basis $\mathcal{B}_1 = \{\be_i,1\le i \le n\}$. We allow $Q$ to have an arbitrary signature 
$-n \le p-q \le n$ where, as usual, $p$ denotes the number of $+1$'s and $q$ denotes the number of $-1$'s in~(\ref{eq:Q}), and $p+q=n$. It is well known, that the equivalence classes of
non-degenerate real quadratic forms under linear transformations are parameterized by the signature. Let $B: V \times V \rightarrow \BR$ be the symmetric bilinear form on~$V$ defined by~$Q$ as 
$B(\bu,\bv) = \frac12(Q(\bu+\bv) - Q(\bu)-Q(\bv))$, and let $\cl_n$ be the universal Clifford algebra of $(V,Q).$ To fix our formalism, we define the associative Clifford product on the exterior algebra 
$\bigwedge V$ through Chevalley's construction~\cite{lounesto}. Namely, first one defines the Clifford product of a vector $\bx \in V$ and a multivector $u \in \bigwedge V$ as 
\begin{align}
\bx u &= \bx \JJ_B u + \bx \wedge u
\label{eq:xu}
\end{align}
where $\bx \JJ_B u$ denotes the \textit{left contraction} of $u$ by $\bx$ with respect to the bilinear form~$B$, and $\bx \wedge u$ is just the wedge product of $\bx$ and $u$ in $\bigwedge V.$ Then, 
one extends this product to two arbitrary multivectors in $\bigwedge V$ by utilizing the following three defining properties of the left contraction:
\begin{enumerate}
\item $\bx \JJ_B \by = B(\bx,\by) = B(\by,\bx)$ (symmetry of the bilinear form) 
\item $\bx \JJ_B (u \wedge v) = (\bx \JJ_B u) \wedge v  + \hat{u} \wedge (\bx \JJ_B v)$,
(Leibniz rule, contractions by grade one elements are derivations)
\item $(u \wedge v ) \JJ_B w = u \JJ_B (v \JJ_B w)$
(left module structure)
\end{enumerate}
for any $\bx,\by \in V$ and $u,v,w \in \bigwedge V.$ Thus, the orthonormal basis $\cb{B}_1$ satisfies the following well-known set of relations in $\cl_n$:
\begin{subequations}
\begin{gather}
\be_i^2 = B(\be_i,\be_i) \cdot 1=Q(\be_i) \cdot 1 = \ve_i \cdot 1,\quad  1 \le i \leq n, 
          \label{eq:B1a}\\
\be_i\be_j + \be_j \be_i = 0, \quad i \neq j,\quad  1 \le i,j \leq n. 
          \label{eq:B1b}
\end{gather}
\label{eq:B1}
\end{subequations}
\hspace*{-1.1ex}
Let $\cb{B}$ be the canonical basis of $\bigwedge V$ generated by $\cb{B}_1.$ That is, let $[n]=\{1,2,\ldots,n\}$ and denote arbitrary, canonically ordered subsets of $[n]$, by underlined Roman characters. The basis elements of $\bigwedge V$, or, of $\cl_n$ due to the linear space isomorphism $\bigwedge V \rightarrow \cl_n$~\cite{lounesto}, can be indexed by these finite ordered subsets as $\be_\iu = \wedge_{i \in \iu}\, \be_i$. Then, an arbitrary element of $\bigwedge V \cong \cl_n$ can be written as
\begin{align}
u &= \sum_{\iu \in 2^{[n]}} u_\iu \be_{\iu}
\label{eq:u}
\end{align}
where $u_\iu \in \BR$ for each $\iu \in 2^{[n]}$ and we identify the unit element $1$ of $\cl_n$ with $\be_\emptyset$. Thus we assume that if $\iu = \{i_1,i_2,\ldots,i_s\}$ for $s=|\iu|$ then 
$\be_\iu = \be_{i_1}\be_{i_2} \cdots \be_{i_s} = \be_{i_1} \wedge \be_{i_2}\wedge  \cdots \wedge \be_{i_s}$ where $i_1 < i_2 < \cdots <i_s$. Therefore, it is better to think of the index set $\iu$ as a list $[i_1, i_2, \cdots,i_s]$ sorted by $<$. Then we declare $\iu = \ju$ if and only if $|\iu|=|\ju|=s$ and $i_1=j_1,i_2=j_2, \ldots,  i_s=j_s$.

Since later we will discuss matrix representations of $\cl_n$, we need to choose a basis for $\cl_n$: Our preferred basis is the exterior algebra basis of simple multivectors $\be_\iu$  (also called \textit{extensors} by Rota \textit{et al.} \cite{rota1994}) sorted by one of four possible \textit{admissible} monomial orders~$\prec$ on $\bigwedge V.$ For a definition and examples of admissible monomial orders in the exterior algebra see~\cite{ablamowicz2009a,GfG} and references therein. We choose for $\prec$ the monomial order called $\mathtt{InvLex}$, or, the \textit{inverse lexicographic order}. For example, when $n=3$, then the basis $\cb{B}$ sorted by $\mathtt{InvLex}$ looks as follows
\begin{equation}
\cb{B}  = \{ 1 \prec \be_1 \prec \be_2 \prec \be_{12} \prec \be_3 \prec \be_{13} \prec \be_{23} \prec \be_{123} \}
\label{eq:sortedB}
\end{equation}
where we have abbreviated $\be_{12}=\be_1 \be_2 = \be_1 \wedge \be_2$, $\be_{123}=\be_1 \be_2 \be_3 = \be_1 \wedge \be_2 \wedge \be_3$, etc., due to the orthogonality of the basis $\cb{B}_1$. The same monomial order has been chosen in~\cite{schottstaples2009} as it allows for a convenient nesting of Clifford algebras, e.g., $\cl_1 \subset \cl_2 \subset \cl_3 \subset \cdots \subset \cl_n$ needed by the authors. The fact that this order is admissible can prove useful later when computing kernels of various operators acting on $\cl_n.$ This is because then the leading terms with respect to $\prec$ of multivector polynomials can easily be identified and methods of non-commutative Gr\"{o}bner bases can be employed. From now on, under the term \textit{sorted basis}, we will understand the standard basis of Grassmann monomials $\cb{B}$ sorted by our chosen monomial order~$\texttt{InvLex}$.

We recall the extension $\lang \cdot,\cdot \rang: \bigwedge V \times \bigwedge V \rightarrow \BR$ of the bilinear form~$B$ to the entire exterior algebra $\bigwedge V$. We will need this symmetric bilinear extension of~$B$ later when we define the Clifford algebra $\cl(V^\ast,Q).$ Following Lounesto~\cite[Chapter 22]{lounesto}, we first extend~$B$ to simple $k$-vectors in 
$\bigwedge^k V$ via
\begin{equation}
\lang \bx_1 \w \bx_2 \w \ldots \w \bx_k,  \by_1 \w \by_2 \w \ldots \w \by_k \rang = 
\det B(\bx_\iu,\by_\ju)  
\end{equation} 
where
\begin{equation}
\det B(\bx_\iu,\by_\ju) = 
\left|
\begin{matrix}    
B(\bx_1,\by_1) & B(\bx_1,\by_2) & \cdots & B(\bx_1,\by_k)\\
B(\bx_2,\by_1) & B(\bx_2,\by_2) & \cdots & B(\bx_2,\by_k)\\
\vdots         & \vdots         & \ddots & \vdots\\
B(\bx_k,\by_1) & B(\bx_k,\by_2) & \cdots & B(\bx_k,\by_k)
\end{matrix}
\right|
\label{eq:detmatrix}
\end{equation}
and further by linearity to all of $\bigwedge^k V$ and by orthogonality to all of $\bigwedge V.$\footnote{In a Hopf algebraic setting this bilinear form resembles a Laplace 
pairing~\cite{fauser-habilitation}.} This last postulate means that a simple $k$-vector is declared orthogonal to any simple $l$-vector relative to the basis  $\cb{B}$ whenever $k \neq l.$ Note that this notion of degree is independent of the choice of a basis, and hence it is well defined. Since our choice of basis $\{\be_i\}$ diagonalizes the polar bilinear form $B$ of $Q$, the matrix of the extended bilinear form $\lang \cdot,\cdot \rang: \bigwedge V \times \bigwedge V \rightarrow \BR$ will be also diagonal in the sorted and also orthonormal basis~$\cb{B}$. For example, for $\cb{B}$ given
in~(\ref{eq:sortedB}) we get the following diagonal matrix:
\begin{align}
\lang \be_\iu,\be_\ju \rang &= \diag (1, \ve_1, \ve_2, \ve_{1}\ve_{2}, \ve_3, 
       \ve_{1}\ve_{3}, \ve_{2}\ve_{3}, \ve_{1}\ve_{2}\ve_{3})
\label{eq:B71}      
\end{align}
where $\be_\iu,\be_\ju \in \cb{B}$. 

Lounesto shows that in the case of the non-degenerate quadratic form $Q$, the left contraction $u \JJ_B v$ of $v \in \bigwedge V$ by $u \in \bigwedge V$ is defined as the dual of the exterior product through the requirement\footnote{In~\cite{fauser2004} this requirement is shown to be the result of product/co-product duality in Grass\-mann-Hopf gebra.} that
\begin{align}
\lang u \JJ_B v, w \rang &= \lang v,\tilde{u} \wedge w \rang \quad \mbox{for all} \; w \in \bigwedge V.
\label{eq:Lduality}
\end{align}
With the help of the three defining properties of the left contraction listed above and the requirement~(\ref{eq:Lduality}), we state and prove the following technical lemma which will be needed later. The proof of this lemma can be found in Appendix~\ref{AppendB}.
\begin{lemma}
Let $\cb{B}$ be the sorted orthonormal basis in $\cl_n=\cl(V,Q)$ for a non-degenerate quadratic form $Q$ defined in~(\ref{eq:Q}). Let $\be_\iu, \be_\ju, \be_\ku$ be any three basis elements in $\cb{B}$ where $\iu,\ju,\ku$ are index lists sorted by $<$. Denote the reversion $(\be_\iu)\tilde{}$ by~$\tilde{\be}_\iu$. The following identities are true:
\begin{itemize}
\item[(i)] Let $\be_\iu = \be_{i_1}\be_{i_2} \cdots \be_{i_s}$ where $s=|\iu|\ge 1.$ Then,
\begin{gather}
\be_\iu \tilde{\be}_\iu 
=\tilde{\be}_\iu \be_\iu 
=(\be_{i_s}\be_{i_{s-1}} \cdots \be_{i_1})(\be_{i_1}\be_{i_2} \cdots \be_{i_s}) 
= \ve_{i_1}\ve_{i_2} \cdots \ve_{i_s} 
\label{eq:idents1}
\end{gather}
\item[(ii)]
\begin{gather}
\lang \be_\iu,\be_\ju  \rang = 
\begin{cases}
0 &\textit{if $\iu \neq \ju$;}\\
1 &\textit{if $\iu=\ju=\emptyset$;}\\
\ve_{i_1}\ve_{i_2} \cdots \ve_{i_s} &\textit{if $\iu = \ju$ and $s=|\iu|\ge 1$}
\end{cases}
\label{eq:idents2}
\end{gather}
\item[(iii)] 
\begin{gather}
\lang \be_\iu, \be_\ju\be_\ku \rang =  
\begin{cases}
\lang \tilde{\be}_\ju \be_\iu, \be_\ku \rang = 0                                             
&\textit{if $\be_\iu \neq \pm \be_\ju\be_\ku$;}\\
\lang \tilde{\be}_\ju \be_\iu, \be_\ku \rang \neq 0 
&\textit{if $\be_\iu = \pm \be_\ju\be_\ku$.}
\end{cases}
\label{eq:idents3}
\end{gather}
\item[(iv)] Let $u,v,w$ be arbitrary multivectors in $\cl_n$. Then we have the following duality formula:
\begin{gather}
\lang u, vw \rang = \lang \tilde{v}u,w \rang 
\label{eq:idents4}
\end{gather}
\end{itemize}
\label{L:lemma1}
\end{lemma}
Formula~(\ref{eq:idents4}) generalizes Lounesto's duality $\lang u \JJ_B v, w \rang = \lang v,\tilde{u} \wedge w \rang$ for all $u,v, w \in \bigwedge V$ to the Clifford algebra $\cl_n$ of the quadratic explicitly non-degenerate form~$Q$. We emphasize here that formula~(\ref{eq:idents4}) is not valid when the form $Q$ is degenerate.

\section{The transposition anti-involution of $\cl_n$}\label{new1}
\medskip

Let $V$ and $Q$ be as in~(\ref{eq:Q}). In particular, recall that in $\cl_n$ we have $\be_i^{-1} \stackrel{\mathrm{def}}{=} (\be_i)^{-1} = \dfrac{\be_i}{\ve_i} = \ve_i \be_i$ since 
$\be_i^2 = \ve_i \cdot 1$ for any $\be_i \in \cb{B}_1$ and $\ve^2=1$ for all $i$.
\begin{definition}
Let $t_\ve: V \rightarrow V$ be the linear map defined, dependent on the signature $\ve$ of~$Q$, as
\begin{equation}
t_\ve(\bx) = t_\ve(\sum_{i=1}^n x_i \be_i) = \sum_{i=1}^n x_i \left(\frac{\be_i}{\ve_i}\right) =
                                             \sum_{i=1}^n x_i \left(\ve_i \be_i\right)
\label{eq:tve} 
\end{equation}
for any $\bx \in V$ and for the orthonormal basis $\cb{B}_1 = \{\be_i, \, 1\leq i \leq n\}$ in $V$ diagonalizing $Q$.
\end{definition}
As the following lemma shows, there are two ways to look at $t_\ve$: (1) As just a linear orthogonal map of $V$; (2) As a \textit{correlation}~\cite{porteous} mapping 
$t_\ve: V \rightarrow V^\ast \cong V$. In the following, we explore both of these 
points of view.
\begin{lemma}
Let $t_\ve$ be the linear map defined in~(\ref{eq:tve}), $\cb{B}_1$ be the orthonormal basis for $V$, and let $\cl_n$ be the universal Clifford algebra of $(V,Q)$.
\begin{itemize}
\item[(i)] $t_\ve$ is an orthogonal involution $V \rightarrow V.$ 
\item[(ii)] The set of vectors $\cb{B}_1^\ast = \{t_\ve(\be_i) = \ve_i \be_i, \, 1\leq i \leq n\}$ gives an orthonormal basis in the dual space $(V^\ast,Q).$
\item[(iii)] Under the identification $V \cong V^\ast$, $t_\ve$ is a symmetric non-degenerate correlation on~$V$ thus making the pair $(V,t_\ve)$ into a \textit{non-degenerate real correlated (linear) space}.
\end{itemize}
\label{l:l1}
\end{lemma}
\begin{proof}
(i) For every $\bx \in V$ we have in $\cl_n$ the following identity: 
\begin{gather*}
Q(t_\ve(\bx)) = t_\ve(\bx)t_\ve(\bx) 
=\sum_{i,j} x_i (\ve_i \be_i) x_j (\ve_j \be_j)
=\sum_i x_i^2 \ve_i^2 \be_i^2 = \sum_i \ve_i x_i^2 = Q(\bx). 
\end{gather*}
It can be easily checked that $t_\ve(t_\ve(\bx)) = \bx, \, \forall \bx \in V.$ Hence, $t_\ve$ is an orthogonal involution on $V.$

(ii) Observe that:
\begin{equation}
\hspace*{2.5ex}\lang t_\ve(\be_i),t_\ve(\be_j) \rang = \lang \ve_i \be_i,\ve_j \be_j \rang = \ve_i \ve_j \lang \be_i, \be_j \rang = \ve_i^2 \delta_{i,j} = 
\begin{cases} 1 & \text{if $i=j$;}\\
              0 & \text{if $i\neq j$}
\end{cases}
\label{eq:dualB1}
\end{equation}
because $\ve_i^2=1, \forall i.$ Thus, the basis $\cb{B}_1^\ast$ is orthonormal with respect to the same quadratic form $Q$. Furthermore, viewing the map $\tve$ as a correlation $V \rightarrow V^\ast$, we can define action of $t_\ve(\bx)\in V^\ast$ on $\by \in V$ for any $\bx \in V$ as
\begin{equation}
t_\ve(\bx)(\by) = \lang t_\ve(\bx), \by \rang.
\label{eq:action1}
\end{equation}
Then we get the expected duality relation among the basis elements in $\cb{B}_1$ and 
$\cb{B}_1^\ast$:
\begin{equation}
t_\ve(\be_i)(\be_j) = \lang \ve_i \be_i, \be_j  \rang = 
\ve_i \lang \be_i, \be_j \rang = \ve_i\ve_j\delta_{i,j} = \delta_{i,j}.
\label{eq:action2}
\end{equation}
Finally, let $\vp = \sum_i \vp(\be_i)t_\ve(\be_i) = \sum_i (\vp_i \ve_i) \be_i \in V^\ast$, where $\vp_i = \vp(\be_i) \in \BR$, be a linear form. Then, under the action~(\ref{eq:action1}), we find the usual result
\begin{equation}
\vp(\bx) = \lang \vp, \bx \rang = 
\sum_{i,j} \vp_i x_j \lang \ve_i \be_i, \be_j \rang =       \sum_i \vp_i x_i.
\label{eq:action3} 
\end{equation}

(iii) From~(\ref{eq:action1}) we get easily that $t_\ve(\bx)(\by) =t_\ve(\by)(\bx)$ for all $\bx,\by \in V$ which means that the correlation $t_\ve$ is symmetric. The rest follows from the fact that the inner product $\lang \cdot,\cdot \rang$ is non-degenerate~\cite{porteous}. 
\end{proof}
We will return to the duality $V \rightarrow V^\ast$ and extend it to the Clifford algebras $\cl(V,Q) \rightarrow \cl(V^\ast,Q)$ in the following section. 

Now we apply Porteous' theorem to the orthogonal involution~$t_\ve$. For now, we take $\bL=\BR$ and $\alpha = 1_\BR$.
\begin{proposition}
Let $A=\cl_n$ be the universal Clifford algebra of $(V,Q)$ and let $t_\ve:V \rightarrow V$ be the orthogonal involution of $V$ defined in~(\ref{eq:tve}). Then there exists a unique algebra involution 
$T_\ve$ of $A$ and a unique algebra anti-involution $\ta{T}{\ve}$ of $A$ such that the following diagrams commute:
\begin{equation}
\bfig
\square|arra|/>` >->` >->`>/[V`V`A`A;t_\ve`\iota`\iota`T_\ve]
\efig
\qquad \mbox{and} \qquad
\bfig
\square|arra|/>` >->` >->`>/[V`V`A`A;t_\ve`\iota`\iota`\ta{T}{\ve}]
\efig
\label{d:diag1}
\end{equation}
In particular, we can define $T_\ve$ and $\ta{T}{\ve}$ as follows:
\begin{itemize}
\item[(i)] For simple $k$-vectors $\be_\iu$ in $\cb{B}$, let $T_\ve(\be_\iu) = T_\ve (\prod_{i \in \iu} \be_i) = \prod_{i \in \iu} t_\ve(\be_i)$ where $k= |\iu|$ and $T_\ve(1_A) = 1_A$. Then, extend 
by linearity to all of $A$.
\item[(ii)] For simple $k$-vectors $\be_\iu$ in $\cb{B}$, let 
\begin{gather}
\ta{T}{\ve}(\be_i) = \ta{T}{\ve} (\prod_{i \in \iu} \be_i) = (\prod_{i \in \iu} t_\ve(\be_i))\tilde{} = (-1)^{\frac{k(k-1)}{2}}\prod_{i \in \iu} t_\ve(\be_i)
\label{eq:defT}
\end{gather} 
where $k=|\iu|$ and $\ta{T}{\ve}(1_A) = 1_A$. Then, extend by linearity to all of $A$.
\end{itemize}
\label{p:p1}
\end{proposition}
\begin{proof}
Due to the uniqueness of $T_\ve$ and $\ta{T}{\ve}$ it is enough to check that the diagrams~(\ref{d:diag1}) commute and that these two maps are, respectively, the involution and the anti-involution induced by $t_\ve$. The latter property follows from the fact that the reversion and $t_\ve$ are commuting involutions. Chasing these diagrams gives for every $\bx \in V$:
\begin{align}
T_\ve(\iota(\bx)) 
  &= T_\ve(\bx) = T_\ve\left(\sum_i x_i \be_i \right)=\sum_i x_i t_\ve(\be_i) 
   = t_\ve(\bx) = \iota(t_\ve(\bx))
\label{eq:T1}  
\end{align} 
where by the abuse of notation we have identified $\bx$ with its image $\iota(\bx)$. Likewise,
\begin{align}
\ta{T}{\ve}(\iota(\bx)) 
  &= \ta{T}{\ve}(\bx) = \ta{T}{\ve}\left(\sum_i x_i \be_i \right)
   = \sum_i x_i (t_\ve(\be_i))\tilde{} = t_\ve(\bx) = \iota(t_\ve(\bx))  
\label{eq:T2}  
\end{align} 
because under reversion $(t_\ve(\be_i))\tilde{} = t_\ve(\be_i), \, \forall i$. 
\end{proof}
Display (\ref{eq:T2}) shows that reversion in the definition of~$\ta{T}{\ve}$ cannot be replaced with conjugation. In Appendix~\ref{AppendA} we show our Maple code of a procedure~$\mathtt{tp}$ which implements the anti-involution $\ta{T}{\ve}$ in~$\cl_n$. 

In the following corollary, we denote the grade involution of the Clifford algebra $\cl_n$ by~$\alpha$, the reversion by~$\beta$, and the conjugation by~$\gamma$.

\begin{corollary}
Let $A = \cl_{p,q}$ and let $T_\ve: A \rightarrow A$ and $\ta{T}{\ve}: A \rightarrow A$ be the involution and the anti-involution of $A$ from Proposition~\ref{p:p1}.
\begin{itemize}
\item[(i)] For the Euclidean signature $(p,q)=(n,0)$, or $p-q=n$, we have $t_\ve = 1_V$. Thus, $T_\ve$ is the identity map $1_A$ on $A$ and $\ta{T}{\ve}$ is the reversion~$\beta$ of $A$.
\item[(ii)] For the anti-Euclidean signature $(p,q)=(0,n)$, or $p-q=-n$, we have $t_\ve = -1_V$. Thus, $T_\ve$ is the grade involution $\alpha$ of $A$ and $\ta{T}{\ve}$ is the conjugation~$\gamma$ of $A$.
\item[(iii)] For all other signatures   $-n < p-q < n$, we have $t_\ve = 1_{V_1} \otimes -1_{V_2}$ where $(V,Q) = (V_1,Q_1) \perp (V_2,Q_2)$. Here, $(V_1,Q_1)$ is the Euclidean subspace of $(V,Q)$ of dimension $p$ spanned by $\{\be_i, \, 1\leq i \leq p \}$ with $Q_1 =Q|_{V_1}$ while $(V_2,Q_2)$ is the anti-Euclidean subspace of $(V,Q)$ of dimension $q$ spanned by 
$\{\be_i, \, p+1\leq i \leq p+q=n \}$ with $Q_2 =Q|_{V_2}$.  Let $A_1 = \cl(V_1,Q_1)$ and $A_2 = \cl(V_2,Q_2)$ so $\cl(V,Q) \cong \cl(V_1,Q_1) \hotimes \cl(V_2,Q_2)$. Thus, 
$$
T_\ve = 1_{A_1} \otimes \alpha_{A_2} \quad \mbox{and} \quad \ta{T}{\ve} = (\beta_{A_1} \otimes\,  \gamma_{A_2}) \, \circ \,(\hat{S}\, \circ\, S)
$$ 
where $S$ is the \emph{ungraded switch} whereas $\hat{S}$ is the \emph{graded switch} defined on $\cl(V_1,Q_1)$ $\hotimes \cl(V_2,Q_2)$.
\item[(iv)] The anti-involution $\ta{T}{\ve}$ is related to the involution $T_\ve$ through the reversion~$\beta$ as follows:
$
\ta{T}{\ve} = T_\ve \circ \beta =  \beta \circ T_\ve. 
$
\end{itemize}
\end{corollary}
\begin{proof}
(i) For Euclidean signatures, it follows from~(\ref{eq:tve}) that $t_\ve = 1_V$. In this case, the identity map $1_A$ is the unique involution on $A$ induced by~$t_\ve$ while the reversion~$\beta$ is the unique anti-involution on $A$ induced by~$t_\ve$. 

(ii) For anti-Euclidean signatures, it follows from~(\ref{eq:tve}) that $t_\ve = -1_V$. In this case, the grade involution $\alpha$ of $A$ is the unique involution on $A$ induced by~$t_\ve$ while the conjugation $\gamma$ is the unique anti-involution on $A$ induced by~$t_\ve$. 

(iii) The orthogonal sum decomposition $(V,Q) = (V_1,Q_1) \perp (V_2,Q_2)$ where $(V_1,Q_1)$ is Euclidean and $(V_2,Q_2)$ is anti-Euclidean follows from the theory of quadratic forms~\cite{lam}. As a consequence, in the category of $\BZ_2$-graded associative algebras, we have the isomorphism 
\begin{gather}
\xy \morphism/ ->/<750,0>[\cl_{p,q}`\cl_{p,0} \hotimes \cl_{0,q};\spl]\endxy
\label{eq:spt}
\end{gather}
That is, $\cl(V_1,Q_1) \hotimes \cl(V_2,Q_2)$ is the \textit{graded} (or \textit{skew}) tensor product of the Clifford algebras~\cite{greub,lam}. Under this identification, we have the following two commutative diagrams:
\begin{equation}
\hspace*{4ex}
\bfig
\square(0,0)|arra|/>`>`>`>/<750,500>%
[\cl_{p,q}`\cl_{p,0} \hotimes \cl_{0,q}`\cl_{p,q}`\cl_{p,0} \hotimes \cl_{0,q};%
 \spl`T_\ve`1_{A_1} \otimes\, \alpha_{A_2}`\spl]
\efig
\quad
\bfig
\square(0,0)|arra|/>`>`>`>/<750,500>%
[\cl_{p,q}`\cl_{p,0} \hotimes \cl_{0,q}`\cl_{p,q}`\cl_{p,0} \hotimes \cl_{0,q};%
 \spl`\ta{T}{\ve}`(\beta_{A_1} \otimes\,  \gamma_{A_2}) \, \circ \,(\hat{S}\, \circ\, S)`\spl]
\efig
\label{eq:diagram12}
\end{equation}
where $S$ is the \textit{ungraded switch} whereas $\hat{S}$ is the \textit{graded switch}~\cite{fauser2004}. Note that we work here in the opposite direction as in the paper~\cite{fauserablam2000}, where we studied the decomposition of (quantum) Clifford algebras. Let\footnote{Here, $\cb{B}_{p,0}$ and $\cb{B}_{0,q}$ denote, respectively, the Grassmann monomial basis in $\cl_{p,0}$ and $\cl_{0,q}$.} 
$\cb{B}_{p,0}$ and~$\cb{B}_{0,q}$ be, respectively, the sorted Grassmann bases for $\cl_{p,0}$ and~$\cl_{0,q}$. The switches are defined on the basis tensors 
$\be_\iu \hotimes \be_\ju \in \cl_{p,0} \hotimes \cl_{0,q}$ for $\be_\iu \in \cb{B}_{p,0}$ and $\be_\ju \in \cb{B}_{0,q}$ as
\begin{align*}
S(\be_\iu \hotimes \be_\ju) = \be_\ju \hotimes \be_\iu \quad \mbox{and} \quad \hat{S}(\be_\iu \hotimes \be_\ju) = (-1)^{|\iu||\ju|}\be_\ju \hotimes \be_\iu.
\end{align*}
Then, their action is extended by linearity to the graded product $\cl_{p,0} \hotimes \cl_{0,q}$. Notice that the combined action of the two switches is obviously
\begin{align*}
(\hat{S} \circ S)(\be_\iu \hotimes \be_\ju) = (-1)^{|\iu||\ju|}\be_\iu \hotimes \be_\ju.
\end{align*}
The extra factor $(-1)^{|\iu||\ju|}$ is needed in the right diagram in~(\ref{eq:diagram12}) due to the fact that the reversion $\beta$ present in $\ta{T}{\ve}$ is an anti-automorphism of  the Clifford algebra $\cl_{p,q}$ dependent on $B$. One can address this sign factor as a bi-character on the grade group. This factor is not needed in the diagram on the left since $T_\ve$ is an automorphism of~$\cl_{p,q}.$ However, the same factor is also implicitly built into the definition of the algebra product of basis monomials in $\cl_{p,0} \hotimes \cl_{0,q}$. We make it clear by explicitly defining the map~(\ref{eq:spt}) 
and showing that it is the required algebra isomorphism. We define the map $\spl$ as follows:
\begin{itemize}
\item[(i)] On the identity element, $\spl(1_{A}) = 1_{A_1} \hotimes 1_{A_2}$.
\item[(ii)] Let $\be_\iu \in \cb{B}$ and let $\iu = \iu_1 \cup \iu_2, \iu_1 \cap \iu_2 =\emptyset$. Then, $\spl(\be_\iu) = \be_{\iu_1} \hotimes \,\be_{\iu_2}$ where $\be_{\iu_1} \in 
\cb{B}_{p,0}$ and $\be_{\iu_2} \in \cb{B}_{0,q}$.
In case $\iu_1$ or $\iu_2$ is empty, we recall that $\be_{\emptyset}=1$.
\item[(iii)] We extend the map $\spl$ by linearity to all elements of $\cl_{p,q}.$
\end{itemize}
When defining the split of $\be_\iu$ into $\be_{\iu_1} \hotimes \,\be_{\iu_2}$, for convenience we relabel the basis elements in $\cb{B}_{0,q}$ modulo $p$. For example, let $p=1$ and $q=2$: 
\begin{multline*}
\cb{B} = [1, \be_1, \be_2, \be_{12},\be_3,\be_{13},\be_{23},\be_{123}] 
\stackrel{\mathrm{split}}{\longmapsto} \\
[1 \hotimes 1, \be_1 \hotimes 1, 1 \hotimes \be_2, \be_{1} \hotimes \be_{1},
1 \hotimes \be_2, \be_{1}   \hotimes \be_{2}, 1 \hotimes \be_{12}, \be_{1} \hotimes \be_{12}]
\end{multline*}
By the reason of dimensionality, namely, $|\cb{B}| = 2^{n} =2^{p+q} = |\cb{B}_{p,0}| |\cb{B}_{0,q}|$, it is clear that the map $\spl$ is a vector space isomorphism. Now we show that it is  the algebra isomorphism. Let $\be_\iu =\be_{\iu_1} \be_{\iu_2}$ and $\be_\ju = \be_{\ju_1} \be_{\ju_2}$ for $\be_{\iu_1},\be_{\ju_1} \in \cb{B}_{p,0}$ and $\be_{\iu_2},\be_{\ju_2} \in \cb{B}_{0,q}$. Then we have the following commutative diagram:
\begin{equation}
\bfig
\square(0,0)|arra|/>`>`>`>/<1500,500>%
[\cl_{p,q} \otimes \cl_{p,q}`(\cl_{p,0} \hotimes \cl_{0,q}) 
           \otimes (\cl_{p,0} \hotimes \cl_{0,q})`
 \cl_{p,q}`\cl_{p,0} \hotimes \cl_{0,q};%
 \spl\, \otimes\, \spl`\cml`\cml`\spl]
\efig
\label{eq:diagram3}
\end{equation}
where the map $\cml$ denotes the Clifford product in $\cl_{p,q}$ and $\cl_{p,0} \hotimes \cl_{0,q}$. When restricted to the basis elements in $\cb{B} \times \cb{B}$, we get when going down and to the right in the above diagram:
\begin{align*}
(\spl \circ \cml) (\be_\iu,\be_\ju) 
&= \spl(\be_\iu \be_\ju) = \spl(\be_{\iu_1}\be_{\iu_2}\be_{\ju_1}\be_{\ju_2})\\
&= (-1)^{|\iu_2||\ju_1|} \spl(\be_{\iu_1}\be_{\ju_1}\be_{\iu_2}\be_{\ju_2})\\
&=(-1)^{|\iu_2||\ju_1|} (\be_{\iu_1}\be_{\ju_1}) \hotimes \, (\be_{\iu_2}\be_{\ju_2})
\end{align*}
where $\be_{\iu_1}\be_{\ju_1} \in \cl_{p,0}$ and $\be_{\iu_2}\be_{\ju_2} \in \cl_{0,q}.$ The factor $(-1)^{|\iu_2||\ju_1|}$ appears due to the defining anticommutation relations on the generators of $\cl_{p,q}$. When going to the right and down we get:
\begin{align*}
(\cml \circ (\spl \times \spl))(\be_\iu,\be_\ju)
&= \cml(\spl(\be_{\iu}), \spl(\be_{\ju}))\\ 
&= \cml(\be_{\iu_1} \hotimes \,\be_{\iu_2},\be_{\ju_1} \hotimes \, \be_{\ju_2})\\
&= (-1)^{|\iu_2||\ju_1|} (\be_{\iu_1}\be_{\ju_1}) \hotimes \, (\be_{\iu_2}\be_{\ju_2})
\end{align*}
where we have used the definition of multiplication of the basis tensors in the product $\cl_{p,0} \hotimes \cl_{0,q}$~\cite{greub,lam,lounesto,porteous}. The factor $(-1)^{|\iu_2||\ju_1|}$ appears now due to this definition.\footnote{See also \cite[Display (2.5)]{fauser2004}.} It can be checked by direct computation that the diagram~(\ref{eq:diagram3}) commutes for any two general elements 
$x,y \in \cl_{p,q},$ or, 
$$
\spl(\cml(x,y)) = \cml(\spl(x),\spl(y)).
$$ 
Thus, the map $\spl$ is an algebra homomorphism, and, so it is an isomorphism. 

The left diagram in~(\ref{eq:diagram12}) commutes due to the equality 
$$
\Hom(\cl_{p,0} \hotimes \cl_{0,q})=\Hom(\cl_{p,0}) \hotimes \Hom(\cl_{0,q}).
$$ 
The mapping shown as the right down arrow is just $T_\ve \hotimes T_\ve$ where $T_\ve$ in the first tensor slot maps $\cl_{p,0} \rightarrow \cl_{p,0}$, hence it reduces to the identity map by part (i). The mapping $T_\ve$ in the second tensor slot maps $\cl_{0,q} \rightarrow \cl_{0,q}$, hence it reduces to the grade involution $\alpha$ by part~(ii). Observe, that since grade involution is an automorphism, no extra sign factor is needed to make this diagram commute.   

The diagram on the right in~(\ref{eq:diagram12}) commutes because any linear mapping 
$$
\chi: \cl_{p,0} \hotimes \cl_{0,q} \longrightarrow \cl_{p,0} \hotimes \cl_{0,q}
$$ 
where $\chi = \sum_k \chi_{1k} \hotimes \chi_{2k}$,  is a tensor product $\varphi \hotimes \psi$ of two linear mappings  $\varphi:\cl_{p,0} \rightarrow \cl_{p,0}$ and 
$\psi:\cl_{0,q} \rightarrow \cl_{0,q}$. This tensor product needs to be modified now by an appropriate sign change on the basis elements since reversion $\beta$ (dependent on $B$) is involved in the definition of $\ta{T}{\ve}$. Let $\be_\iu \in \cb{B}$ and like before we write $\be_\iu = \be_{\iu_1} \be_{\iu_2}$ where  $\be_{\iu_1} \in \cb{B}_{p,0}$ and $\be_{\iu_2} \in \cb{B}_{0,q}$. Consider the following commutative diagram:
\begin{equation}
\bfig
\square(0,0)|arra|/>`>`>`>/<1000,500>%
[\cl_{p,q}`\cl_{p,0} \hotimes \cl_{0,q}`
 \cl_{p,q}`\cl_{p,0} \hotimes \cl_{0,q};%
 \spl `\beta`(\beta_p \otimes \beta_q) \circ (\hat{S} \circ S)`\spl]
\efig
\label{eq:diagram4}
\end{equation}
where $\beta$ is the reversion in $\cl_{p,q}$ dependent on the bilinear form~$B$, $\beta_p$ is the reversion in $\cl_{p,0}$ dependent on the bilinear form~$B_p$ defined by~$Q_1$, and $\beta_q$ is the reversion in $\cl_{0,q}$ dependent on the bilinear form~$B_q$ defined by~$Q_2$. When restricted to the basis elements in $\cb{B}$, we get when going down and to the right in the above diagram:
\begin{align*}
(\spl \circ \beta)(\be_\iu) 
&= \spl (\beta(\be_{\iu_1}\be_{\iu_2})) = \spl (\beta(\be_{\iu_2}) \beta(\be_{\iu_1}))\\
&= \spl ((-1)^{|\iu_1||\iu_2|} \beta(\be_{\iu_1}) \beta(\be_{\iu_2})) 
= (-1)^{|\iu_1||\iu_2|} \beta_p(\be_{\iu_1}) \hotimes \beta_q(\be_{\iu_2})
\end{align*}
for any $\be_{\iu} \in \cb{B}$. Likewise, when going to the right and then down, we get
\begin{align*}
((\beta_p \otimes \beta_q) &\circ (\hat{S} \circ S) \circ \spl )(\be_\iu)\\[0.5ex]
&=(\beta_p \otimes \beta_q) \circ (\hat{S} \circ S)( \be_{\iu_1} \hotimes \, \be_{\iu_2} )
=(\beta_p \otimes \beta_q) \circ \hat{S} (\be_{\iu_2} \hotimes \, \be_{\iu_1})\\[0.5ex]
&=(\beta_p \otimes \beta_q) (-1)^{|\iu_1||\iu_2|} (\be_{\iu_1} \hotimes \, \be_{\iu_2})
=(-1)^{|\iu_1||\iu_2|} \beta_p(\be_{\iu_1}) \hotimes \, \beta_q(\be_{\iu_2})
\end{align*}

Thus, the diagram~(\ref{eq:diagram4}) commutes when input is restricted to the basis monomials in $\cb{B}.$ By direct computation and using linearity one can show that this diagram in fact commutes for any general element $u \in \cl_{p,q}$.

From the commutativity of the left diagram in~(\ref{eq:diagram12}) and the diagram~(\ref{eq:diagram4}) as well as the fact that $\ta{T}{\ve} = T_\ve \circ \beta,$ we obtain commutativity of the right diagram in~(\ref{eq:diagram12}). In fact,  the right down arrow in that diagram is  
\begin{gather*}
(\ta{T}{\ve} \otimes \ta{T}{\ve}) \circ (\hat{S} \circ S) = ((\beta_p \circ 1_{A_1}) \otimes (\beta_q \circ \alpha_{A_2})) \circ (\hat{S} \circ S)
=(\beta_{A_1} \otimes \gamma_{A_2}) \circ (\hat{S} \circ S)
\end{gather*}
where as before $A_1 = \cl_{p,0}$ and $A_2 = \cl_{0,q}$.

(iv) This follows directly from the definitions of $T_\ve$ and $\ta{T}{\ve}$ given in Proposition~\ref{p:p1}. On the basis elements of $\cb{B}$ we get:
\begin{gather*}
\ta{T}{\ve}(\be_\iu) = \ta{T}{\ve} (\prod_{i \in \iu} \be_i) = 
\beta(\prod_{i \in \iu} t_\ve(\be_i)) = \beta (T_\ve(\be_\iu))
                   = (-1)^{\frac{k(k-1)}{2}} T_\ve(\be_\iu) = T_\ve(\beta(\be_\iu))
\end{gather*}
and the rest follows due to the linearity of $T_\ve$.
\end{proof}
In the case of the diagonal form $B$ used in this paper, the Clifford product of the basis generators coincides with their exterior product. This is because the generators are orthogonal so 
$\be_i\be_j = \be_i \wedge \be_j$  whenever $i\neq j$, etc.. Thus, the Grassmann basis $\cb{B}^\wedge$ for $\cl_n \cong \bigwedge V$ consisting of $1$ and Grassmann monomials, e.g., 
$\be_{i_1} \wedge \be_{i_2} \wedge \cdots \wedge \be_{i_s}$ coincides with the Clifford basis $\cb{B}$ while in general these two bases are \textit{different}. It turns out that the involution 
$T_\ve: \cl_n \rightarrow \cl_n$ and the anti-involution $\ta{T}{\ve}: \cl_n \rightarrow \cl_n$ happen to be also, respectively, the involution and the anti-involution of $\bigwedge V$. Thus, for completeness, we state without proof the following corollary.   
 
\begin{corollary}
Let $T_\ve: \cl_n \rightarrow \cl_n$ and $\ta{T}{\ve}: \cl_n \rightarrow \cl_n$ be the involution and the anti-involution of $\cl_n$ from Proposition~\ref{p:p1}. Then, $T_\ve$ is an involution of 
$\bigwedge V$ and~$\ta{T}{\ve}$ is an anti-involution of $\bigwedge V$.
\end{corollary}

\section{Clifford algebra over the dual space}
\label{dual}
Since $(V^\ast,Q)$ is a non-degenerate quadratic space spanned by the orthonormal basis $\cb{B}^\ast_1$ from Lemma~\ref{l:l1}, part (iii), we can define the Clifford algebra $\cl(V^\ast,Q)$ as expected. 
\begin{definition}
The \textit{Clifford algebra over the dual space} $V^\ast$ is the universal Clifford algebra $\cl(V^\ast,Q)$ of the quadratic pair $(V^\ast,Q)$. For short, we denote this algebra by $\cl_n^\ast$.
\label{def:dualcliff}
\end{definition}

\begin{remark}
Although from now on we denote the Clifford algebra of the dual $(V^\ast,Q)$ via $\cl_n^\ast$, we do not claim that $\cl_n^\ast$ is the \textit{dual} algebra of $\cl_n$ in categorical sense as it was considered in~\cite{fauserstumpf1997} and references therein. Connection between $\cl(V^\ast,Q)$ and the dual $(\cl(V,Q))^\ast$ of the Clifford algebra $\cl(V,Q)$ will be investigated elsewhere.
\end{remark}

Thus, the orthonormal basis $\cb{B}^\ast_1$ satisfies the same relations in $\cl^\ast_n$ as the basis $\cb{B}_1$ satisfies in $\cl_n$ namely
\begin{subequations}   
\begin{gather}
(t_\ve(\be_i))^2 = B(t_\ve(\be_i),t_\ve(\be_i)) \cdot 1=Q(t_\ve(\be_i)) \cdot 1 = \ve_i \cdot 1,\quad  1 \le i \leq n,
         \label{eq:B1star1}\\
t_\ve(\be_i)t_\ve(\be_j) + t_\ve(\be_j) t_\ve(\be_i) = 0, \quad i \neq j,\quad  1 \le i,j \leq n.
         \label{eq:B1star2}
\end{gather}
\label{eq:B1star}
\end{subequations}
\hspace*{-1.8ex}
We denote by $\cb{B}^\ast$ the canonical basis of $\bigwedge V^\ast \cong \cl^\ast_n$ generated by $\cb{B}^\ast_1$ and sorted by~$\mathtt{InvLex}$. That is, we define
$\cb{B}^\ast = \{T_\ve(\be_\iu) \,| \, \be_\iu \in \cb{B}\}$ given that
\begin{equation}
\lang T_\ve(\be_\iu),\be_\ju \rang = \delta_{\iu,\ju} 
\label{eq:dualbasisClstar}
\end{equation}
for $\be_\iu,\be_\ju \in \cb{B}$ and $T_\ve(\be_\iu) \in \cb{B}^\ast$. For example, when $n=3$, then
\begin{multline}
\cb{B}^\ast  = \{ 1 \prec \ve_1\be_1 \prec \ve_2\be_2 \prec 
                  \ve_1\ve_2\be_{12} \prec \ve_3\be_3 \prec 
                  \ve_1\ve_3 \be_{13} \prec \ve_2\ve_3\be_{23} \prec 
                  \ve_1\ve_2\ve_3\be_{123} \}
\label{eq:sortedBstar}
\end{multline}
Then, an arbitrary linear form $\vp$ in 
$\bigwedge V^\ast \cong \cl^\ast_n$ can be written as
\begin{equation}
\vp = \sum_{\iu \in 2^{[n]}} \vp_\iu T_\ve(\be_\iu)
\label{eq:vp}
\end{equation}
where $\vp_\iu \in \BR$ for each $\iu \in 2^{[n]}$. Due to the linear isomorphisms $V \cong V^\ast$ and $\bigwedge V^\ast \cong \cl(V^\ast,Q)$, we extend, by a small abuse of notation, the inner product $\lang \cdot,\cdot \rang$ defined in Section~2 to 
\begin{equation}
\lang \cdot,\cdot \rang: \bigwedge V^\ast \times \bigwedge V^\ast \rightarrow \BR.
\label{eq:dualprod}
\end{equation}
This way we find, as expected, that the matrix of this inner product on $\bigwedge V^\ast$ is also diagonal, that is, that the basis~$\cb{B}^\ast$ is orthonormal with respect to 
$\lang \cdot, \cdot \rang$. For example, for~$\cb{B}^\ast$ given in~(\ref{eq:sortedBstar}) we get the same matrix as in~(\ref{eq:B71}), namely,
\begin{equation}
\lang T_\ve(\be_\iu),T_\ve(\be_\ju) \rang = \diag (1, \ve_1, \ve_2, \ve_{1}\ve_{2}, \ve_3, 
       \ve_{1}\ve_{3}, \ve_{2}\ve_{3}, \ve_{1}\ve_{2}\ve_{3})
\label{eq:B72}      
\end{equation}
where $T_\ve(\be_\iu),T_\ve(\be_\ju) \in \cb{B}^\ast$. We extend the action of dual vectors from $V^\ast$ on $V$ to all linear forms $\varphi$ in $\cl^\ast_n$ acting on multivectors $v$ in $\cl_n$ via the inner product~(\ref{eq:dualprod}) as
\begin{equation}
\vp(v) = \lang \vp,v \rang = \sum_{\iu \in 2^{[n]}} \vp_\iu v_\iu
\label{eq:phiaction}
\end{equation}
given that 
$$
\vp = \sum_{\iu \in 2^{[n]}} \vp_\iu T_\ve(\be_\iu) \in \cl^\ast_n 
\quad \mbox{where} \quad 
\vp_\iu = \vp(\be_\iu) \in \BR
$$ 
and $v = \sum_{\iu \in 2^{[n]}} v_\iu \be_\iu  \in \cl_n$ for some coefficients $v_\iu \in \BR$.

We recall the definition of the \textit{transpose} of a linear mapping~\cite{lang}.
\begin{definition}
Let $T:V \rightarrow U$ be an arbitrary linear mapping from a $k$-vector space~$V$ into a $k$-vector space~$U$. Now for any functional $\vp \in U^\ast$, the composition $\vp \circ T$ is a linear mapping from $V$ to $k$:
\begin{equation}
\bfig
\qtriangle|alr|/>`>`>/[V`U`k;T`\vp \circ T`\vp]
\efig
\label{eq:Tphidiag}
\end{equation}
That is, $\vp \circ T \in V^\ast$. Thus, the correspondence $\vp \mapsto \vp \circ T$ is a mapping from $U^\ast$ to $V^\ast$. We denote it by $T^t$ and call it the \textit{transpose} of $T$. That is, 
$T^t: U^\ast \rightarrow V^\ast$ is defined by $T^t(\vp) = \vp \circ T$ and $(T^t(\vp))(v) =  \vp(T(v))$ for every $v \in V$.
\label{def:transpose}
\end{definition}

The following two facts about the mapping $T^t$ are standard in linear algebra and can be found, for example, in~\cite{lipschutz}:
\begin{theorem} 
Let $T:V \rightarrow U$ be a linear mapping and let $T^t: U^\ast \rightarrow V^\ast$ be the transpose of~$T$.
\begin{itemize}
\item[(i)] $T^t$ is linear. 
\item[(ii)] Let $M$ be the matrix representation of $T$ relative to bases $\{v_i\}$ in $V$ and $\{u_j\}$ in $U$. Then the transpose matrix $M^t$ is the matrix representation of the transpose 
$T^t: U^\ast \rightarrow V^\ast$ relative to the bases dual to $\{u_j\}$ and $\{v_i\}$.
\end{itemize}
\label{t:fromlang}
\end{theorem}
\noindent
Since we are interested in matrices of the left regular representation of $\cl_n$, we define the following left multiplication operator on $\cl_n$.
\begin{definition}
Let $u$ be an arbitrary multivector in $\cl_n$. Then, the \textit{left multiplication} operator $L_u$ is simply the map $L_u : \cl_n \rightarrow \cl_n, \, v \mapsto uv, \, \forall v \in \cl_n.$
\label{def:Lu} 
\end{definition}
\noindent
We have the following proposition.
\begin{proposition}
Let $\cb{B}$ and $\cb{B}^\ast$ be, respectively, the sorted bases in $\cl_n$ and $\cl^\ast_n$, and let $L_u:\cl_n \rightarrow \cl_n$, $u \in \cl_n$, be the left multiplication operator.
\begin{itemize}
\item[(i)] The operator $L_{\tilde{u}}$ is the dual of $L_u$ with respect to the inner product $\lang \cdot,\cdot \rang: \bigwedge V \times \bigwedge V \rightarrow \BR$ defined in 
Section~\ref{definitions}. That is, 
\begin{gather}
\lang u, L_v(w) \rang = \lang L_{\tilde{v}}(u),w \rang
\label{eq:res1}
\end{gather} 
for any $v,w \in \cl_n$.
\item[(ii)] Let $u,v \in \cl_n$ and $\vp \in \cl_n^\ast$. Then,
\begin{subequations}
\begin{equation}
L_u^t(\vp)(v) = (\vp \circ L_u)(v)  = \vp(L_u(v)) = L_{\tilde{u}}(\vp)(v) 
\label{eq:res2a}
\end{equation}
or, equivalently,
\begin{equation}
\lang L_u^t(\vp),v   \rang = \lang \vp \circ L_u, v   \rang = 
\lang \vp, L_u(v)\rang = \lang L_{\tilde{u}}(\vp),v   \rang
\label{eq:res2b}
\end{equation}
\label{eq:res2}
\end{subequations}
\item[(iii)] Let $T_\ve$ and $\ta{T}{\ve}$ be as above, $\vp \in \cl^\ast_n \cong \cl_n$, and $u \in \cl_n$. Then,
\begin{equation}
L_{\tilde{u}}(\vp) = T_{\ve}(L_{\ta{T}{\ve}(u)}(\vp_{\cb{B}})) = 
(\ta{T}{\ve}(L_{\ta{T}{\ve}(u)}(\vp_{\cb{B}})))\tilde{}
\label{eq:res3}
\end{equation}
where $\vp_{\cb{B}} = T_\ve(\vp)$ is the form $\vp$ expressed in the $\cb{B}$ basis.
\item[(iv)] Let $u \in \cl_n$. If $\matrepr{L_u}$ is the matrix of the operator $L_u$ relative to the basis $\cb{B}$ and $\matrepr{L_{\tilde{u}}}$ is the matrix of the operator $L_{\tilde{u}}$ relative to the basis  $\cb{B}^\ast$, then 
\begin{equation}
\matrepr{L_u}^T = \matrepr{L_{\tilde{u}}}
\label{eq:res4}
\end{equation}
where $\matrepr{L_u}^T$ is the matrix transpose of $\matrepr{L_u}$.
\item[(v)] Let $u \in \cl_n$. If $\matrepr{L_u}$ is the matrix of the operator $L_u$ relative to the basis $\cb{B}$ and $\matrepr{L_{\ta{T}{\ve}(u)}}$ is the matrix of the operator $L_{\ta{T}{\ve}(u)}$ relative to the basis $\cb{B}$, then 
\begin{equation}
\matrepr{L_u}^T = \matrepr{L_{\ta{T}{\ve}(u)}} = \matrepr{  L_{T_\ve(\tilde{u})}   } 
\label{eq:res5}
\end{equation}
where $\matrepr{L_u}^T$ is the matrix transpose of $\matrepr{L_u}$. 
\item[(vi)] Let $u \in \cl_n$.  The anti-involution $\ta{T}{\ve}$ applied to $u$ results in the transposition of the matrix $\matrepr{L_u}$ in the left regular representation 
$L_u: \cl_n \rightarrow \cl_n$ relative to the basis~$\cb{B}$.
\end{itemize}
\label{p:main}
\end{proposition}
\begin{proof}
For easier reading, in the following we will use interchangeably $\beta(u)$ and $\beta(\vp)$ to denote, respectively, the reversion $\tilde{u}$ of $u \in \cl_n$ and the reversion $\tilde{\vp}$ of  
$\vp \in \cl^\ast_n$. Recall that reversion is an anti-involution of $\cl_n$, $\cl^\ast_n$ and $\bigwedge V$.

(i) This just follows directly from~(\ref{eq:idents4}) and the definition of $L_u$.

(ii) We apply~(\ref{eq:Tphidiag}) to the linear mapping $L_u:\cl_n \rightarrow \cl_n$ and get, for every $\vp \in \cl_n^\ast$, the following commutative diagram:
\begin{equation}
\bfig
\qtriangle|alr|/>`>`>/[\cl_n`\cl_n`k;L_u`L_u^t (\vp) = \vp \,\circ\, L_u`\vp]
\efig
\label{eq:Luphidiag}
\end{equation}
where $L_u^t = (L_u)^t$ is the transpose of $L_u$, that is, $L_u^t : \cl_n^\ast \rightarrow \cl_n^\ast$ is the correspondence $\vp \mapsto \vp \circ L_u$. We show now that 
$L_u^t(\vp) = L_{\tilde{u}}(\vp)$ for every $\vp \in \cl_n^\ast$. 

Let $v$ be an arbitrary multivector in $\cl_n$. Then,
\begin{align}
L_u^t(\vp)(v) = (\vp \circ L_u)(v) 
&=\vp(uv) = \lang \vp,uv \rang =\lang \tilde{u}\vp, v \rang \notag \\
&= \lang L_{\tilde{u}}(\vp),v \rang = L_{\tilde{u}}(\vp)(v), \quad \forall v \in \cl_n 
\end{align}
where we have used the definition of the action of the multiform $\vp \in \cl_n$ on~$v$ given in~(\ref{eq:phiaction}), formula~(\ref{eq:idents4}), and the linear isomorphism $\cl^\ast_n \cong \cl_n$. Thus, $L_u^t(\vp) = L_{\tilde{u}}(\vp)$ which is~(\ref{eq:res2a}). Display~(\ref{eq:res2b}) is just a reformulation of~(\ref{eq:res2a}) in terms of the inner product. 

(iii) We use the fact that $T_\ve \circ \beta = \beta \circ T_\ve$ where $T_\ve$ is the involution of $\cl_n$. Recall also that $\vp_{\cb{B}} = T_\ve(\vp)$  and 
$\ta{T}{\ve} = T_\ve \circ \beta = \beta \circ T_\ve$. 
Then we get the following:
\begin{align*}
L_{\tilde{u}}(\vp)
&= \tilde{u}\vp =  T_\ve(T_\ve(\tilde{u}\vp)) =T_\ve(T_\ve(\tilde{u})T_\ve(\vp)) \notag \\
&= T_\ve(\ta{T}{\ve}(u) \vp_{\cb{B}}) =T_\ve(L_{\ta{T}{\ve}(u)}(\vp_{\cb{B}})) =
   (\ta{T}{\ve}   (L_{\ta{T}{\ve}(u)}(\vp_{\cb{B}})))\tilde{}
\label{eq:der1}
\end{align*}

(iv) We apply the second part of Theorem~\ref{t:fromlang}. From~(\ref{eq:res2}) above we get that the transpose $L_u^t$ of $L_u$ is $L_{\tilde{u}}$, that is, 
$L_u^t(\vp) = L_{\tilde{u}}(\vp)$ for every $\vp \in \cl_n^\ast$. This means, that the matrix of $L_u$ relative to the basis  $\cb{B}$ of $\cl_n$ is the transpose of the matrix of 
$L_{\tilde{u}}$ relative to the basis  $\cb{B}^\ast$ of $\cl_n^\ast$. This is precisely the condition~(\ref{eq:res4}).

(v) We combine relation $L_u^t(\vp) = L_{\tilde{u}}(\vp), \forall \vp \in \cl_n^\ast$, from~(\ref{eq:res2}) with relation $L_{\tilde{u}}(\vp) = T_{\ve}(L_{\ta{T}{\ve}(u)}(\vp_{\cb{B}}))$ 
from~(\ref{eq:res3}). Recall that $\vp = T_\ve(\vp_{\cb{B}})$ where  $\vp_{\cb{B}}$ is the multiform $\vp$ expressed relative to the basis  $\cb{B}$ of~$\cl_n$. Then we get the following sequence of equalities:
\begin{align*}
L_u^t(\vp) = L_{\tilde{u}}(\vp) &= T_{\ve}(L_{\ta{T}{\ve}(u)}(\vp_{\cb{B}})),\notag\\
(L_u^t \circ T_\ve)(\vp_\cb{B}) &= T_{\ve}(L_{\ta{T}{\ve}(u)}(\vp_{\cb{B}})),\notag\\
(T_\ve \circ L_u^t \circ T_\ve)(\vp_\cb{B}) &= L_{\ta{T}{\ve}(u)}(\vp_{\cb{B}}) 
\end{align*}
where the last equality means that the matrix of $L_u$ relative to the basis  
$\cb{B}$ of $\cl_n$ is the transpose of the matrix of $L_{\ta{T}{\ve}(u)}$ relative to the basis  $\cb{B}$ of $\cl_n$. This is because the latter is the matrix of the transpose  $L_u^t$ of $L_u$ in the $\cb{B}$ basis.

(vi) This is just a restatement of part (v) and formula~(\ref{eq:res5}).
\end{proof}

\section{Examples}
\label{examples}
In this section we provide a few examples to illustrate parts (iv) and (v) of Proposition~\ref{p:main}. In the first example we verify relation~(\ref{eq:res4}) in low dimensions.
\begin{example}($n=1$) Let $\cb{B} = \{1, \be_1\}$ be the sorted basis for $\cl_1$. Then, 
$\cb{B}^\ast = \{1, \ve_1\be_1\}$ is the sorted basis for the Clifford algebra 
$\cl_1^\ast$. Let $u=a_1 1+a_2 \be_1$, $a_1,a_2 \in \BR$, be an arbitrary element in $\cl_1$. Then, matrix $\matrepr{L_u}$ of the operator $L_u$ relative to $\cb{B}$ and matrix $\matrepr{L_{\tilde{u}}}$ of the operator $L_{\tilde{u}}$ relative to $\cb{B}^\ast$ are
\begin{equation}
\matrepr{L_u} =\left[\begin{matrix} a_1 & \ve_1 a_2\\a_2 & a_1 \end{matrix}\right] 
\quad \mbox{and} \quad
\matrepr{L_{\tilde{u}}} =\left[\begin{matrix} a_1 & a_2\\ \ve_1 a_2 & a_1 \end{matrix}\right].
\label{eq:example1a}
\end{equation}
Hence, as expected, they are related by the transposition.

We repeat this computation for $n=2$. Then, $\cb{B} = \{1,\be_1,\be_2,\be_{12}\}$ is the sorted basis for $\cl_2$ and $\cb{B}^\ast = \{1,\ve_1\be_1,\ve_2\be_2,\ve_1\ve_2\be_{12}\}$ is the sorted basis for $\cl_2^\ast$. Let $u = a_1 1 + a_2 \be_1 + a_3 \be_2 + a_4 \be_{12}$ be a general element in $\cl_2$. The reversed element is then $\tilde{u} = a_1 1 + a_2 \be_1 + a_3 \be_2 - a_4 \be_{12}$ and we easily find matrices for $L_{u}$, relative to $\cb{B},$ and $L_{\tilde{u}}$, relative to $\cb{B}^\ast$, as:
\begin{equation}
\matrepr{L_u} =\left[\begin{matrix}%
a_1 & \ve_1 a_2  & \ve_2 a_3 & -\ve_1\ve_2 a_4 \\
a_2 & a_1        & \ve_2 a_4 & -\ve_2 a_3\\ 
a_3 & -\ve_1 a_4 & a_1       & \ve_1 a_2 \\
a_4 & -a_3       & a_2       & a_1%
\end{matrix}\right],
\label{eq:example1b1}
\end{equation}
and
\begin{equation}
\matrepr{L_{\tilde{u}}} =\left[\begin{matrix}%
a_1 & a_2  & a_3 & a_4 \\
\ve_1 a_2 & a_1      & -\ve_1  a_4 & -a_3\\ 
\ve_2 a_3 & \ve_2 a_4 & a_1       & a_2 \\
-\ve_1 \ve_2 a_4 & -\ve_2 a_3     & \ve_1 a_2       & a_1%
\end{matrix}\right]. 
\label{eq:example1b2}
\end{equation}
Matrices displayed in~(\ref{eq:example1a}) appear as two diagonal blocks in, respectively, matrices~(\ref{eq:example1b1}) and~(\ref{eq:example1b2}). 

Finally, let $n=3$. Then, in $\cl_3$ we choose the sorted basis $\cb{B}$ displayed in~(\ref{eq:sortedB}) whereas the sorted basis $\cb{B}^\ast$ displayed in~(\ref{eq:sortedBstar}) is our  basis for the Clifford algebra $\cl_3^\ast$.  Let $u = a_1 1 +a_2 \be_1+a_3 \be_2+a_4 \be_{12}+a_5 \be_3+a_6 \be_{13}+a_7 \be_{23}+a_8 \be_{123}$ be a general element in $\cl_3$. Then, the matrix of $L_u$ is 
\begin{equation}
\matrepr{L_u} =\left[\begin{matrix}%
a_{1} & \ve_1 a_{2}& \ve_2 a_{3}&  -\ve_{12} a_{4}& \ve_3 a_{5}& -\ve_{13}a_{6} 
      &  -\ve_{23} a_{7}&  -\ve_{123} a_{8}\\
a_{2} & a_{1} & \ve_2 a_{4}&  -\ve_2 a_{3}& \ve_3 a_{6}&  -\ve_3 a_{5}&  -\ve_{23}a_{8} 
      &  -\ve_{23} a_{7}\\
a_{3} &  -\ve_1 a_{4}& a_{1} & \ve_{1} a_{2}& \ve_3 a_{7}& \ve_{13} a_{8}&  -\ve_3a_{5} 
      & \ve_{13} a_{6}\\
a_{4} &  -a_{3} & a_{2} & a_{1} & \ve_3 a_{8}& \ve_3 a_{7}&  -\ve_3 a_{6}& \ve_3 a_{5}\\
a_{5} &  -\ve_1 a_{6}&  -\ve_2 a_{7}&  -\ve_{12} a_{8}& a_{1} & \ve_1 a_{2}
      & \ve_2 a_{3}&  -\ve_{12} a_{4}\\
a_{6} &  - a_{5} &  -\ve_2 a_{8}&  -\ve_2 a_{7}& a_{2} & a_{1} & \ve_2a_{4} 
      &  -\ve_2 a_{3}\\
a_{7} & \ve_1 a_{8}&  -a_{5} & \ve_1 a_{6}& a_{3} &  -\ve_1 a_{4}& a_{1} 
      & \ve_1 a_{2}\\
a_{8} & a_{7} &  -a_{6} & a_{5} & a_{4} &  - a_{3} & a_{2} & a_{1}
\end{matrix}\right]
\label{eq:example1c1}
\end{equation}
where to shorten the output we have set $\ve_{12}=\ve_1\ve_2,$ $\ve_{123}=\ve_1\ve_2\ve_3$, etc. The matrix of $\matrepr{L_{\tilde{u}}}$ is again the matrix transpose of~(\ref{eq:example1c1}) which can be verified by direct computation. Furthermore, again due to our \texttt{InvLex} sorting of $\cb{B}$, matrix~(\ref{eq:example1b1}) appears twice on the diagonal of~(\ref{eq:example1c1}).
\label{ex:one}
\end{example}
It can be easily observed from the form of the matrices $\matrepr{L_u}$ in dimensions $n=1,2,3$ that the lower dimensional matrix always appears \textit{twice} on the diagonal of the next higher dimensional matrix. This should be true for all dimensions. Using \texttt{CLIFFORD}, \cite{ablamfauser2009} we have verified  relation~(\ref{eq:res4}) in dimensions $1$ through $8.$

If we define a mapping $L:\cl_n \rightarrow \End (\cl_n),\, u\mapsto L_u,$ where $L_u$ is the left multiplication operator from Definition~\ref{def:Lu}, then $L$ is an algebra 
\textit{homomorphism} referred to as the \textit{left regular representation}. Therefore, rather than finding matrix $\matrepr{L_u}$ for an arbitrary element $ u \in \cl_n$, it is of course sufficient to find matrices that represent the basis generators $\be_1,\be_2,\ldots,\be_n \in 
\cb{B}_1$. This is because 
$$
\matrepr{L_{\be_\iu}} = \matrepr{L_{\be_{i_1}}}\matrepr{L_{\be_{i_2}}} \cdots 
                        \matrepr{L_{\be_{i_s}}} 
\quad  \mbox{whenever} \quad
\be_\iu = \be_{i_1} \be_{i_2} \cdots \be_{i_s} \in \cb{B}.
$$
Thus, $\matrepr{L_u} = \sum_\iu a_{\iu} \matrepr{L_{\be_{\iu}}}$ whenever $u =\sum_\iu a_{\iu} \be_\iu \in \cl_n$. It is interesting to observe the block structure of these generic matrices 
$\matrepr{L_{\be_{i}}}$ relative to the sorted basis $\cb{B}$ which in this paper coincides with the standard basis $\cb{B}^\wedge$ for the exterior algebra~$\bigwedge V$. 

To simplify our notation, we let $\matrepr{\be_{i}} \equiv \matrepr{L_{\be_{i}}}$. Also, by $ M \oplus N$ we mean the block-diagonal matrix 
$\left[\begin{array}{c|c} M & 0\\ \hline 0 & N\end{array}\right]$ which we abbreviate to $2M$ if $M=N$, etc.. Furthermore, let $J_n$ be the $2^{n-1} \times 2^{n-1}$ matrix defined recursively as 
$$
J_n = J_{n-1} \oplus -J_{n-1} \quad \mbox{with} \quad 
J_1 = \left[\begin{matrix} 1 \end{matrix} \right].
$$ 

\begin{example}
Let $\cb{B}_1 =\{\be_1,\ldots,\be_n\}$ be as before the set of $1$-vector generators for~$\cl_n$. Denote the $2^n \times 2^n$ matrix $\matrepr{L_{\be_{i}}}$ by $\E{n}{i}$ for all $1 \le i \le n$. Then, either by direct computation or from~(\ref{eq:example1a}), (\ref{eq:example1b1}), and (\ref{eq:example1c1}), we get the following matrices representing the $1$-vector generators in the left regular representation relative to~$\cb{B}$.
\begin{itemize}
\item[(i)] When $n=1,$ we get
$$
\E{1}{1} =\left[\begin{array}{c|c} 0 & \ve_1 J_1\\ \hline 1 & 0  \end{array}\right] = 
          \left[\begin{matrix} 0 & \ve_1\\1 & 0 \end{matrix}\right].
$$
\item[(ii)] When $n=2,$ we get
$$
\E{2}{1} = 2\E{1}{1}=\left[\begin{matrix} 0 &\ve_1 & 0 & 0\\
                                          1 & 0    & 0 & 0\\
                                          0 & 0    & 0 & \ve_1\\
                                          0 & 0    & 1 & 0  
                     \end{matrix}\right], 
\quad 
\E{2}{2} = \left[\begin{array}{c|c} 0 & \ve_2 J_2\\ \hline J_2 & 0\end{array}\right]=
                     \left[\begin{matrix} 0 & 0 & \ve_2 & 0\\
                                          0 & 0 & 0 & -\ve_2\\
                                          1 & 0  & 0 & 0\\
                                          0 & -1 & 0 & 0  
                     \end{matrix}\right].
$$  
\item[(iii)] When $n=3,$ we get
$$
\E{3}{1} = 2\E{2}{1}, \qquad 
\E{3}{2} = 2\E{2}{2}, \qquad
\E{3}{3} = \left[\begin{array}{c|c} 0 &  \ve_3 J_3 \\ \hline 
                                    J_3 &  0%
           \end{array}\right].
$$    
\item[(iv)] When $n=4,$ we get
$$
\E{4}{1} = 2\E{3}{1}, \qquad
\E{4}{2} = 2\E{3}{2}, \qquad
\E{4}{3} = 2\E{3}{3}, \qquad
\E{4}{4} = \left[\begin{array}{c|c}  0 &  \ve_4 J_4 \\ \hline 
                                    J_4 & 0%
           \end{array}\right].
$$
\end{itemize}
From the above, we can easily conjecture the form of all matrices $\E{n}{i}, \,1\le i \le n,$ for $n \ge 2$, via the following recursive relation:
\begin{equation}
\E{n}{1} = 2\E{n-1}{1},\;
\cdots, \;
\E{n}{n-1} = 2\E{n-1}{n-1},\;
\E{n}{n} = \left[\begin{array}{c|c}  0 &  \ve_n J_n \\ \hline 
                                     J_n & 0%
           \end{array}\right], \; 
\E{1}{1} = \left[\begin{matrix} 0 & \ve_1\\1 & 0 \end{matrix}\right].
\label{eq:recE}
\end{equation}
This recursive relation has been checked with \texttt{CLIFFORD} \cite{ablamfauser2009} for $n \le 6$. Finally, for each $n,$ these matrices of course satisfy the necessary relations:
$$
(\E{n}{i})^2 = \ve_i \mathbf{1} \quad \mbox{and} \quad \E{n}{i}\E{n}{j} = -\E{n}{j}\E{n}{i} 
                 \quad \mbox{for} \quad 1 \le i,j \le n \quad \mbox{and} \quad i\neq j
$$
where $\mathbf{1}$ is the $2^n \times 2^n$ identity matrix.
\label{ex:two}
\end{example}
Using now result~(\ref{eq:res5}), we can easily find matrices in the left regular representation $L$ that represent the dual $1$-vectors $\ta{T}{\ve}(\be_i)=\ve_i\be_i \in \cb{B}_1^\ast$.

\begin{example}
Let $\cb{B}_1^\ast =\{\ve_1\be_1,\ldots,\ve_n\be_n\}$ be as before the set of $1$-vector generators for~$\cl_n^\ast$. Denote the $2^n \times 2^n$ matrix $\matrepr{L_{\ta{T}{\ve}(\be_i)}}$ by $\F{n}{i}$ for all $1 \le i \le n$. Then, either by direct computation or from~(\ref{eq:res5}) and Example~\ref{ex:two}, we get the following matrices representing the $1$-vector generators in the left regular representation  relative to~$\cb{B}$.
\begin{itemize}
\item[(i)] When $n=1,$ we get
$$
\F{1}{1} =\left[\begin{array}{c|c} 0 & 1\\ \hline \ve_1 J_1 & 0  \end{array}\right] = 
          \left[\begin{matrix} 0 & 1\\\ve_1 & 0 \end{matrix}\right].
$$
\item[(ii)] When $n=2,$ we get
$$
\F{2}{1} = 2\F{1}{1}=\left[\begin{matrix} 0 & 1 & 0 & 0\\
                                          \ve_1 & 0    & 0 & 0\\
                                          0 & 0    & 0 & 1\\
                                          0 & 0    & \ve_1 & 0  
                     \end{matrix}\right], 
\quad 
\F{2}{2} = \left[\begin{array}{c|c} 0 & J_2\\ \hline \ve_2 J_2 & 0\end{array}\right]=
                     \left[\begin{matrix} 0 & 0 & 1 & 0\\
                                          0 & 0 & 0 & -1\\
                                          \ve_2 & 0  & 0 & 0\\
                                          0 & -\ve_2 & 0 & 0  
                     \end{matrix}\right],
$$  
\end{itemize}
and so on. In general, for an arbitrary $n$, we get the matrix transpose of matrices in~(\ref{eq:recE}):
\begin{equation}
\F{n}{1} = 2\F{n-1}{1},\;
\cdots, \;
\F{n}{n-1} = 2\F{n-1}{n-1},\;
\F{n}{n} = \left[\begin{array}{c|c}  0 &  J_n \\ \hline 
                                     \ve_n J_n & 0%
           \end{array}\right], \; 
\F{1}{1} = \left[\begin{matrix} 0 & 1\\\ve_1 & 0 \end{matrix}\right].
\label{eq:recF}
\end{equation}
\label{ex:three}
\end{example}

Given that for a specific signature $(p,q)$, the unity of the universal Clifford algebra $\cl_{p,q}$ has a finite decomposition into a sum of mutually annihilating primitive idempotents, the Clifford algebra itself is a direct sum of left (or right) irreducible $\cl_{p,q}$-modules of the form $\cl_{p,q}f$ for some \textit{primitive} idempotent~$f.$ Thus, in general, the left regular representation of $\cl_n$ is reducible and can be expressed as a direct sum of irreducible (spinor) representations.\footnote{This is a special feature of quadratic spaces. If a Clifford algebra is defined for a general non-symmetric bilinear form such a decomposition into irreducibles is no longer possible, see~\cite{fauserablam2000,fauser-vaccum}.} Obviously, since each spinor representation is the restriction of the left regular representation to a spinor module (minimal left ideal) of the form $\cl_{p,q}f$, relation~(\ref{eq:res4}) remains true for the spinor representation in the case when $f\cl_{p,q}f \cong \BR,$ that is, when $p-q = 0,1,2 \bmod 8$ and the algebra  $\cl_{p,q}$ is simple, i.e., when $p-q \neq 1 \bmod 4$. We illustrate this in the next example whereas the general case of spinor representation is treated in the next section.

\begin{example}
Consider $\cl_{1,1} \cong \Mat(2,\BR)$ with the decomposition of the unity $1 = f_1 + f_2$ where $f_1= \frac12(1+\be_{12})$ and $f_2= \frac12(1-\be_{12})$ are two primitive mutually annihilating idempotents. Define $S_1$ and $S_2$ as the following two minimal (spinor) left ideals $S_1 = \cl_{1,1}f_1 = \spn_\BR\{f_1, \be_1f_1 \}$ and $S_2 = \cl_{1,1}f_2 = \spn_\BR\{f_2, \be_1f_2 \}.$ Then, 
$\cl_{1,1} = S_1 \oplus S_2$ is the decomposition of $\cl_{1,1}$ as the left $\cl_{1,1}$-module. Let $u = a_1 1+a_2 \be_1+a_3 \be_2+a_4\be_{12} \in \cl_{1,1}$. The restrictions of the left regular representation $L_u |_{S_i}, i=1,2,$ are irreducible and faithful. In particular, we easily find that
$$
\matrepr{L_u |_{S_1}} = 
\left[\begin{matrix}a_1+a_4 & a_2-a_3\\a_2 + a_3 & a_1 - a_4 \end{matrix} \right], \quad
\matrepr{L_u |_{S_2}} = 
\left[\begin{matrix}a_1 -a_4 & a_2+a_3\\a_2 - a_3 & a_1 + a_4 \end{matrix} \right].
$$
Furthermore, as expected, the left regular representation $L_u$ decomposes into two irreducible components:
$$
\matrepr{L_u} =  \matrepr{L_u |_{S_1}} \oplus  \matrepr{L_u |_{S_2}} = 
\left[\begin{array}{c|c} \matrepr{L_u |_{S_1}} &  0 \\ \hline 
                         0 & \matrepr{L_u |_{S_2}} \end{array} \right].
$$
Using result~(\ref{eq:res5}), we confirm by direct computation that matrices $\matrepr{L_u |_{S_1}}$ and $\matrepr{L_{\ta{T}{\ve}(u)} |_{S_1}}$ are related via the transposition since
$$
\matrepr{L_{\ta{T}{\ve}(u)} |_{S_1}} = 
\left[\begin{matrix}a_1+a_4 & a_2+a_3\\a_2 - a_3 & a_1 - a_4 \end{matrix} \right]
$$
where $\ta{T}{\ve}(u) = a_1 1 +a_2 \be_1-a_3 \be_2+a_4 \be_{12}$.
\label{ex:four}
\end{example}

We repeat this last example for Clifford algebras of signatures $(4,2)$ and $(1,3)$ which are important in describing the symmetry group of the hydrogen atom, the AdS space, and the symmetries of the Minkowski space, respectively.

\begin{example}
We consider $\cl_{4,2} \cong \Mat(8,\BR)$ with the decomposition of unity $1 = \sum_{i=1}^8 f_i$ where we let $f_1 = \frac18 (1+\be_1)(1+\be_{35})(1+\be_{46})$ and the remaining seven primitive idempotents are obtained from $f_1$ by independently changing the plus signs to minus signs. Then, the following list of eight basis monomials
$$
\cb{M} = [1,\be_2,\be_3,\be_4,\be_{23},\be_{24},\be_{34},\be_{234}]
$$
can be selected to generate spinor ideals 
$$
S_i = \cl_{4,2}f_i = \spn_\BR \{1f_i,\be_2 f_i,\be_3f_i,\be_4f_i,\be_{23}f_i,\be_{24}f_i,
\be_{34}f_i,\be_{234}f_i \}
$$ 
for $i=1,\ldots,8.$ Then, $\cl_{4,2} = \oplus_{i=1}^8 S_i$ is the direct decomposition of $\cl_{4,2}$ as the left $\cl_{4,2}$-module. Rather than finding a matrix of an arbitrary element 
$u \in \cl_{4,2}$ in the left regular representation $L_u: S_i \rightarrow S_i$ for some $i$, due to the linearity of the representation we find matrices $\matrepr{L_{\be_k}}$ for $k=1,\ldots,8.$ Knowing these matrices is sufficient to find $\matrepr{L_u}$ for any $u \in \cl_{4,2}.$ Then, we find matrices $\matrepr{L_{ \ta{T}{\ve}(\be_k)}} = \matrepr{L_{ t_\ve(\be_k)}}$  and realize that they are related to the former via the transposition.

Recall, that $t_\ve(\be_k) = \be_k$ for $k=1,\ldots,4$ and $t_\ve(\be_k) = -\be_k$ for $k=5,6$. This means that matrices $\matrepr{L_{\be_k}}$ and $\matrepr{L_{ t_\ve(\be_k)}}$ for $k=1,\ldots,5$
are obviously identical and symmetric, whereas these matrices for $k=5,6$ will differ by an overall sign. Then, we obtain:  
\begin{multline*}
\matrepr{L_{\be_1}},\matrepr{L_{ t_\ve(\be_1)}} = \\ 
\left[ 
{\begin{matrix}
1 & 0 & 0 & 0 & 0 & 0 & 0 & 0 \\
0 & -1 & 0 & 0 & 0 & 0 & 0 & 0 \\
0 & 0 & -1 & 0 & 0 & 0 & 0 & 0 \\
0 & 0 & 0 & -1 & 0 & 0 & 0 & 0 \\
0 & 0 & 0 & 0 & 1 & 0 & 0 & 0 \\
0 & 0 & 0 & 0 & 0 & 1 & 0 & 0 \\
0 & 0 & 0 & 0 & 0 & 0 & 1 & 0 \\
0 & 0 & 0 & 0 & 0 & 0 & 0 & -1
\end{matrix}}
 \right] , \, \left[ 
{\begin{matrix}
1 & 0 & 0 & 0 & 0 & 0 & 0 & 0 \\
0 & -1 & 0 & 0 & 0 & 0 & 0 & 0 \\
0 & 0 & -1 & 0 & 0 & 0 & 0 & 0 \\
0 & 0 & 0 & -1 & 0 & 0 & 0 & 0 \\
0 & 0 & 0 & 0 & 1 & 0 & 0 & 0 \\
0 & 0 & 0 & 0 & 0 & 1 & 0 & 0 \\
0 & 0 & 0 & 0 & 0 & 0 & 1 & 0 \\
0 & 0 & 0 & 0 & 0 & 0 & 0 & -1
\end{matrix}}
 \right] 
\end{multline*}
\begin{multline*}
\matrepr{L_{\be_2}},\matrepr{L_{ t_\ve(\be_2)}} = \\ 
\left[ 
{\begin{matrix}
0 & 1 & 0 & 0 & 0 & 0 & 0 & 0 \\
1 & 0 & 0 & 0 & 0 & 0 & 0 & 0 \\
0 & 0 & 0 & 0 & 1 & 0 & 0 & 0 \\
0 & 0 & 0 & 0 & 0 & 1 & 0 & 0 \\
0 & 0 & 1 & 0 & 0 & 0 & 0 & 0 \\
0 & 0 & 0 & 1 & 0 & 0 & 0 & 0 \\
0 & 0 & 0 & 0 & 0 & 0 & 0 & 1 \\
0 & 0 & 0 & 0 & 0 & 0 & 1 & 0
\end{matrix}}
 \right] , \, \left[ 
{\begin{matrix}
0 & 1 & 0 & 0 & 0 & 0 & 0 & 0 \\
1 & 0 & 0 & 0 & 0 & 0 & 0 & 0 \\
0 & 0 & 0 & 0 & 1 & 0 & 0 & 0 \\
0 & 0 & 0 & 0 & 0 & 1 & 0 & 0 \\
0 & 0 & 1 & 0 & 0 & 0 & 0 & 0 \\
0 & 0 & 0 & 1 & 0 & 0 & 0 & 0 \\
0 & 0 & 0 & 0 & 0 & 0 & 0 & 1 \\
0 & 0 & 0 & 0 & 0 & 0 & 1 & 0
\end{matrix}}
 \right] 
\end{multline*}
\begin{multline*}
\matrepr{L_{\be_3}},\matrepr{L_{ t_\ve(\be_3)}} = \\ 
\left[ 
{\begin{matrix}
0 & 0 & 1 & 0 & 0 & 0 & 0 & 0 \\
0 & 0 & 0 & 0 & -1 & 0 & 0 & 0 \\
1 & 0 & 0 & 0 & 0 & 0 & 0 & 0 \\
0 & 0 & 0 & 0 & 0 & 0 & 1 & 0 \\
0 & -1 & 0 & 0 & 0 & 0 & 0 & 0 \\
0 & 0 & 0 & 0 & 0 & 0 & 0 & -1 \\
0 & 0 & 0 & 1 & 0 & 0 & 0 & 0 \\
0 & 0 & 0 & 0 & 0 & -1 & 0 & 0
\end{matrix}}
 \right] , \, \left[ 
{\begin{matrix}
0 & 0 & 1 & 0 & 0 & 0 & 0 & 0 \\
0 & 0 & 0 & 0 & -1 & 0 & 0 & 0 \\
1 & 0 & 0 & 0 & 0 & 0 & 0 & 0 \\
0 & 0 & 0 & 0 & 0 & 0 & 1 & 0 \\
0 & -1 & 0 & 0 & 0 & 0 & 0 & 0 \\
0 & 0 & 0 & 0 & 0 & 0 & 0 & -1 \\
0 & 0 & 0 & 1 & 0 & 0 & 0 & 0 \\
0 & 0 & 0 & 0 & 0 & -1 & 0 & 0
\end{matrix}}
 \right] 
\end{multline*}
\begin{multline*}
\matrepr{L_{\be_4}},\matrepr{L_{ t_\ve(\be_4)}} = \\ 
\left[ 
{\begin{matrix}
0 & 0 & 0 & 1 & 0 & 0 & 0 & 0 \\
0 & 0 & 0 & 0 & 0 & -1 & 0 & 0 \\
0 & 0 & 0 & 0 & 0 & 0 & -1 & 0 \\
1 & 0 & 0 & 0 & 0 & 0 & 0 & 0 \\
0 & 0 & 0 & 0 & 0 & 0 & 0 & 1 \\
0 & -1 & 0 & 0 & 0 & 0 & 0 & 0 \\
0 & 0 & -1 & 0 & 0 & 0 & 0 & 0 \\
0 & 0 & 0 & 0 & 1 & 0 & 0 & 0
\end{matrix}}
 \right] , \, \left[ 
{\begin{matrix}
0 & 0 & 0 & 1 & 0 & 0 & 0 & 0 \\
0 & 0 & 0 & 0 & 0 & -1 & 0 & 0 \\
0 & 0 & 0 & 0 & 0 & 0 & -1 & 0 \\
1 & 0 & 0 & 0 & 0 & 0 & 0 & 0 \\
0 & 0 & 0 & 0 & 0 & 0 & 0 & 1 \\
0 & -1 & 0 & 0 & 0 & 0 & 0 & 0 \\
0 & 0 & -1 & 0 & 0 & 0 & 0 & 0 \\
0 & 0 & 0 & 0 & 1 & 0 & 0 & 0
\end{matrix}}
 \right] 
\end{multline*}
\begin{multline*}
\matrepr{L_{\be_5}},\matrepr{L_{ t_\ve(\be_5)}} = \\ 
\left[ 
{\begin{matrix}
0 & 0 & -1 & 0 & 0 & 0 & 0 & 0 \\
0 & 0 & 0 & 0 & 1 & 0 & 0 & 0 \\
1 & 0 & 0 & 0 & 0 & 0 & 0 & 0 \\
0 & 0 & 0 & 0 & 0 & 0 & -1 & 0 \\
0 & -1 & 0 & 0 & 0 & 0 & 0 & 0 \\
0 & 0 & 0 & 0 & 0 & 0 & 0 & 1 \\
0 & 0 & 0 & 1 & 0 & 0 & 0 & 0 \\
0 & 0 & 0 & 0 & 0 & -1 & 0 & 0
\end{matrix}}
 \right] , \, \left[ 
{\begin{matrix}
0 & 0 & 1 & 0 & 0 & 0 & 0 & 0 \\
0 & 0 & 0 & 0 & -1 & 0 & 0 & 0 \\
-1 & 0 & 0 & 0 & 0 & 0 & 0 & 0 \\
0 & 0 & 0 & 0 & 0 & 0 & 1 & 0 \\
0 & 1 & 0 & 0 & 0 & 0 & 0 & 0 \\
0 & 0 & 0 & 0 & 0 & 0 & 0 & -1 \\
0 & 0 & 0 & -1 & 0 & 0 & 0 & 0 \\
0 & 0 & 0 & 0 & 0 & 1 & 0 & 0
\end{matrix}}
 \right] 
\end{multline*}
\begin{multline*}
\matrepr{L_{\be_6}},\matrepr{L_{ t_\ve(\be_6)}} = \\ 
\left[ 
{\begin{matrix}
0 & 0 & 0 & -1 & 0 & 0 & 0 & 0 \\
0 & 0 & 0 & 0 & 0 & 1 & 0 & 0 \\
0 & 0 & 0 & 0 & 0 & 0 & 1 & 0 \\
1 & 0 & 0 & 0 & 0 & 0 & 0 & 0 \\
0 & 0 & 0 & 0 & 0 & 0 & 0 & -1 \\
0 & -1 & 0 & 0 & 0 & 0 & 0 & 0 \\
0 & 0 & -1 & 0 & 0 & 0 & 0 & 0 \\
0 & 0 & 0 & 0 & 1 & 0 & 0 & 0
\end{matrix}}
 \right] , \, \left[ 
{\begin{matrix}
0 & 0 & 0 & 1 & 0 & 0 & 0 & 0 \\
0 & 0 & 0 & 0 & 0 & -1 & 0 & 0 \\
0 & 0 & 0 & 0 & 0 & 0 & -1 & 0 \\
-1 & 0 & 0 & 0 & 0 & 0 & 0 & 0 \\
0 & 0 & 0 & 0 & 0 & 0 & 0 & 1 \\
0 & 1 & 0 & 0 & 0 & 0 & 0 & 0 \\
0 & 0 & 1 & 0 & 0 & 0 & 0 & 0 \\
0 & 0 & 0 & 0 & -1 & 0 & 0 & 0
\end{matrix}}
 \right] 
\end{multline*}
\end{example}
\normalsize

\section{Conclusions}
Anti-involutions, such as the transposition and Hermitian conjugation in linear algebra, play an important role in many constructions and classification arguments on linear spaces. Such involutive maps are well studied and well understood. A further complication is added if one considers quadratic vector spaces $V \in \mathbf{QVect}_\BR$, where on every vector space a (non degenerate) quadratic form is available $Q : V \longrightarrow \mathbb{R}$. Such spaces come thereby equipped with a (real) Clifford algebra $\cl(\mathbb{R}^{p+q},Q_{p,q})$. Nondegenerate quadratic forms over the reals are classified by their signature $\ve$ with $\ve_i \in \mathbb{R}^*/\mathbb{R}^2$. Minimal faithful representations of the real Clifford algebras are realized in spinor spaces over a (double) field $\BK$. Via the spinor spaces, Clifford algebras are isomorphic to (full) matrix algebras $\mathrm{Mat}(r_{p,q}\times r_{p,q},\BK)$, where $r_{p,q}$ is the dimension given by the Radon-Hurwitz number and $\BK$ is a (double, possibly skew) real field. This setting leads to the question how the natural transposition anti-automorphism on the matrix representation can be pulled back into the Clifford algebra framework. Intriguingly enough, the transposition splits into two parts, the matrix transposition on the matrix representation on the spinor modules and also, dependent on the signature (classification) of the quadratic form, an involution (conjugation) or anti-involution (reversion for the skew field) on the base field $\BK$ of the spinor modules. In this first paper we have thus established an extension of the reversion anti-involution to any signature.

A quadratic form $Q$ on a real vector space can be described via a polarization
\begin{align}
\label{polarization}
\xymatrix{
Q~:~ V \ar@{->}[r]^{\Delta}
& V\otimes V 
\ar@{->}[r]^{B}
& \mathbb{R}
}
\end{align}
where $\Delta : V \rightarrow V\otimes V$ is the diagonal map, and $B : V\otimes V \rightarrow \mathbb{R}$ is a (necessarily symmetric) bilinear form (here assumed to be non-degenerate). Since $(V,Q)$ is described by a (real) symmetric matrix, there exists an orthonormal basis diagonalizing  $B = S^t\circ D \circ S$, such that for the eigenvectors $\{v_i\}$ of $B$ one finds $e_i = S v_i$ for the canonical basis $\{e_i\}$ of $V$.

If one is not only interested in the points of the space $V$, but also in linear subspaces, one deals with the Grassmann algebra $\bigwedge V$ or antisymmetric forms. The Clifford algebra $\cl(V,B)$ is a subalgebra of the endomorphisms algebra of the Grassmann algebra $\cl(V,B) \subset \mathrm{End}(\bigwedge V) \cong \bigwedge V  \otimes \bigwedge V^*$. The last isomorphism relies on the dual isomorphism $* : V \rightarrow V^*$ canonically obtained via the dual basis $\{e^*_i\}$ with $e^*_i(e_j)=\delta_{i,j}$. 

However, the (non-degenerate) bilinear form $B$ induces a second identification of $V$ and $V^*$ via $\flat : V \rightarrow V^*$ and $V^* \ni v^\flat = B(v,\cdot)$, of course depending on the (class of) the quadratic form $Q$ via its polarization $B$. Note that the Clifford algebra extends the bilinear form $B$ to a bilinear form $B^\wedge : \bigwedge V\otimes \bigwedge V \rightarrow \mathbb{R}$. This allows to define a dual isomorphism on the underlying vector space (induced by $\flat$ extended to $\bigwedge V$) which we called dual Clifford algebra $\cl^*(V,B)$. The map implementing this dual isomorphism in the chosen orthonormal basis $\{e_i\}$ is just the map  $T^{\tilde{~}}_\varepsilon \equiv \mathtt{tp}$, which depends on the signature $(p,q)$ of the quadratic form $Q$ and which has been studied in this paper. This map is an anti-involution, and the paper showed at length that it induces the transposition on the matrix representations 
$\mathrm{Mat}(r_{p,q}\times r_{p,q},\BK)$ in case when $\BK \cong \BR$, i.e., when $(p-q)=0,1,2 \mod 8.$ In  \cite{part2} we study all remaining cases (including semisimple algebras) and show that $\tp$ induces also Hermitian complex conjugation or Hermitian quaternionic conjugation in spinor representation. 

In \cite{part2}, we will study another aspect, namely the matter of stabilizer groups of primitive idempotents in the Clifford algebra. Such stabilizer groups allow to construct bases for the isomorphic spinor ideals and are hence important to provide isomorphisms between them. Moreover, such groups are used to model physical theories. 

In the following, we focus on possibilities opened by the present work for further applications and research.
\begin{list}{$\bullet$}{\setlength{\leftmargin}{.15in}
                        \setlength{\itemsep}{1.0ex plus 0.2ex minus 0ex}}
\item  Transposition is essential for singular value decomposition, which opens a way to study spectral properties of intertwining maps $U : V \rightarrow W$ in the category 
$\mathbf{FdVect}_\BR$, \emph{including} their implications on spaces of linear subspaces (Grassmannians), see~\cite{ablamowicz2005,fauser2006}.
\item  Spinor spaces emerge as representation spaces of Clifford algebras, and hence bear an important structure for such frameworks as particle physics or quantum computing. Important, w.r.t. the present work, is how the quadratic form $Q$ induces a spinor space inner product. The paper elucidates this in the real case and shows how the signature $(p,q)$ is relevant for this purpose. While
this was well known, our work exhibits a complete classification of all cases for all signatures (possibly for the first time). This enables one to study the relation between null spaces (and quadrics related to them) and the respective spinor inner products. In \cite{part2} we show, that the Hermitian complex conjugation and the Hermitian quaternionic conjugations emerge from the transposition of the
real Clifford algebra. This may help one understand categorical frameworks in the foundations of quantum computing, where such structures are derived from so-called dagger-structures, see \cite{vicary}.
\item  We have avoided stating general results about representations of Clifford algebras $\cl_n$ with $n>9$. However, our result show clearly recursive structures which have to be exploited. This is usually done by using so called periodicity theorems, like $\cl_{p,q} = \cl_{q-1,p-1}\otimes \cl_{1,1}$ or the famous $\bmod \, 8$ periodicity. While such periodicity theorems are readily
obtained in the case of symmetric bilinear forms, its not clear how they look in the general case~\cite{fauserablam2000}. Finally, one would like to keep also the link to the stabilizer groups, and this work is beyond the aim of the present paper.
\item  The construction of primitive idempotents, done in the second paper \cite{part2}, which project out the minimal left/right ideals, parallels the construction of Young idempotents in the more general setting of representation theory. A further development of the present work should clarify the relation to this construction and to the theory of highest weights. This is needed to formalize the problem described in the previous paragraph. 
\item Clifford algebras are related to spin and pin groups. These groups emerge as groups of units subjected to further algebraic conditions~\cite{salingaros1,salingaros2,salingaros3,varlamov}, such
as the requirement that the algebraic norm $x{\tilde{x}}=1$ for invertible $x\in \cl(V,B)$. Having alternative anti-involutions allows obviously to define new types of such groups via $x\tp(x)=1$, see also work of Chisholm-Farwell~\cite{ChisholmFarwell} and Schmeikal~\cite{schmeikal1} and references therein.
\item Group algebras, such as the ones obtained by the Salingaros vee groups, are usually Frobenius algebras. That is, there exists an isomorphism between left $G$-modules and right $G$-modules. In other words, the parastrophic matrix, a contraction of the group multiplication table (a rank $3$ matrix) with a vector (yielding the parastrophic matrix of rank $2$), is invertible for at least one non-zero contraction vector. Clifford algebras are very closely related to Frobenius algebras. The present work opens the opportunity to investigate this relationship making it more clear what kind of additional
structure a Frobenius algebra has to have to be Cliffordian. This opens way to many important problems related to the cohomology of flag manifolds of the Grassmannians (Schubert calculus, Schubert and Grothendieck polynomials~\cite{lenart}) to solutions of systems having Markovian traces (like the Jones polynomial) and Yang-Baxter equations and ring theoretic questions, see~\cite{kadison}. This research is related to the next point.
\item The present work prompts a question, what would happen if we would not use the diagonal map in (\ref{polarization}), but a general coproduct, a question already asked by 
Helmstetter~\cite{helmstetterMonoids,helmstetterArbQuadForms,helmstetterCliffordWeyl}. One is lead to a category of vector spaces with arbitrary not necessarily symmetric bilinear forms. One is able to construct Clifford algebras for these spaces either, and they were baptized Quantum Clifford Algebras in~\cite{fauser-habilitation}. The antisymmetric part introduces a twist in the dual isomorphism used to identify the Grassmann algebras $\bigwedge V$ and $\bigwedge V^*$ making this a $\mathbb{Z}_2$-graded morphism but inhomogeneous otherwise. This can be seen as a trivial Cliffordization and is useful, e.g., in  quantum field theory~\cite{fauser-wick-paper}. The antisymmetric part in the bilinear form can be used to model $q$-deformations and $q$-spinors, 
see~\cite{fauser-JPhysA,fauser-groupXXII,ablamowicz-fauser}. However, the present work allows now a much more concise study of such generalized framework. For example, one can re-introduce recursively a $\BZ$-grading along the lines discussed  in~\cite[Section8B]{hahn}. In particular, we expect to find generalizations of the Salingaros vee group algebras to Salingaros vee Hopf algebras. Such Hopf algebras can be related to anyons, but may also be important as spinor like covers of other well known Hopf algebras emerging from bilinear generalizations of quadratic spaces. 
\end{list}

\section*{Acknowledgments}
Bertfried Fauser would like to thank the Emmy-Noether Zentrum for Algebra at the University of Erlangen for their hospitality during his stay at the Zentrum in 2008/2009.
\appendix
\section{Code of the transposition procedure $\mathtt{tp}$}
\label{AppendA}

Here is the code of the procedure ${\tt tp}$ that accomplishes the `transposition' anti-involution $\tp$  in $\cl_{p,q}$. This procedure requires $\mathtt{CLIFFORD}$ package~\cite{ablamfauser2009}. It was first presented in~\cite{ablamowicz2005}.

{\small
\begin{tabbing}
\; \=\; \=\; \=\; \=\; \=\; \=\; \=\; \=\;\kill
\keyw{tp}\texttt{:=proc(xx::\{clibasmon,climon,clipolynom\}) local x,L,p,co,u:}\\
\> \texttt{x:=Clifford:-displayid(xx):}\\
\> \keyw{if} \texttt{type(x,clibasmon)} \keyw{then}\\
\>\>  \keyw{if} \texttt{x=Id} \keyw{then} \keyw{return} \texttt{Id} \keyw{else}\\ 
\>\>\>  \texttt{p:=op(Clifford:-cliterms(x));}\\
\>\>\>  \texttt{L:=Clifford:-extract(p,'integers'):} \texttt{\#list L of indices}\\
\>\>\>  \texttt{L:=[seq(L[nops(L)-i+1],i=1..nops(L))];} \texttt{\#reversed list L}\\
\>\>\>  \texttt{u:=Clifford:-cmul(seq(B[L[i],L[i]]*cat(e,L[i]),i=1..nops(L)));}\\
\>\>\>  \keyw{return} \texttt{Clifford:-reorder(u)}\\
\>\> \keyw{end if}:\\
\> \keyw{elif} \texttt{type(x,climon)} \keyw{then}\\
\>\> \texttt{p:=op(Clifford:-cliterms(x)):}\\
\>\> \texttt{co:=coeff(x,p);}\\
\>\> \keyw{return} \texttt{co*}\keyw{procname}\texttt{(p)}\\
\> \keyw{elif} \texttt{type(x,clipolynom)} \keyw{then}\\
\>\> \keyw{return} \texttt{Clifford:-clilinear(x,}\keyw{procname}\texttt{)}\\
\> \keyw{end if}:\\  
\keyw{end tp}:\\
\end{tabbing}
} 

\section{Proof of Lemma~\ref{L:lemma1}}
\label{AppendB}
\begin{proof}
(i) This identity follows easily from the associativity of the Clifford product, the properties of the reversion, and the fact that $\be_i \be_i = \ve_i$ for any $\be_i \in \cb{B}_1.$

(ii) If $\iu \neq \ju$ then either $|\iu|\neq |\ju|$ or the lists are of the same length but their contents are different. In the first case, $\be_\iu$ and $\be_\ju$ are orthogonal because their grades are different. In the second case, $\exists i_t \in \iu$ such that $i_t \notin \ju$. This means that the $t$-th  row in the determinant~(\ref{eq:detmatrix}) is zero.

Let $\be_1 \in V$ be such that $B(\be_1,\be_1)=\be_1 \JJ_B \be_1= \ve_1 = \pm 1$. Then, 
\begin{align*}
\lang 1,1 \rang 
= \lang \ve_1 ,\ve_1 \rang = \lang \be_1 \JJ_B \be_1,\ve_1  \rang = \lang \be_1, \tilde{\be}_1 \wedge \ve_1 \rang 
= \ve_1 \lang \be_1,\be_1 \rang = \ve_1^2 = 1.
\end{align*}

Assume now that $\iu = \ju$ as ordered lists and  
$
\be_\iu = \be_\ju = \be_{i_1}\be_{i_2} \cdots \be_{i_s} = 
\be_{i_1}\wedge \be_{i_2}\wedge  \cdots \wedge \be_{i_s} 
$ 
where $s=|\iu|\ge 1$. Then, $\lang \be_{i_1}, \be_{i_1}\rang = \ve_{i_1}.$ So, by induction on~$s$, assume that the third case formula in~(\ref{eq:idents2}) is true for $s-1$ and compute:
\begin{align}
\lang \be_\iu,\be_\iu \rang%
&= \lang \be_{i_1}\wedge \be_{i_2}\wedge  \cdots \wedge \be_{i_s},
   \be_{i_1}\wedge \be_{i_2}\wedge  \cdots \wedge \be_{i_s}\rang  \notag\\
&= \lang \be_{i_1} \JJ_B (\be_{i_1}\wedge \be_{i_2}\wedge  \cdots \wedge \be_{i_s}),
   \be_{i_2}\wedge  \cdots \wedge \be_{i_s} \rang \notag\\            
&= \lang  (\be_{i_1} \JJ_B \be_{i_1}) \wedge  (\be_{i_2}\wedge  \cdots \wedge \be_{i_s}),
   \be_{i_2}\wedge  \cdots \wedge \be_{i_s} \rang  \notag\\
&= \ve_{i_1} \lang \be_{i_2}\wedge  \cdots \wedge \be_{i_s},
   \be_{i_2}\wedge  \cdots \wedge \be_{i_s} \rang = \ve_{i_1}\ve_{i_2} \cdots \ve_{i_s}  
\end{align} 
where in the third line we have used the fact that $\be_{i_1} \JJ_B (\be_{i_2}\wedge  \cdots \wedge \be_{i_s})=0$ since for every $i_t >i_1$ we have $\be_{i_1} \JJ_B \be_{i_t} = 0$.

(iii) Observe first that due to the orthonormality of $\cb{B}$ the product of two basis monomials $\be_\ju,\be_\ku$ is, up to a scalar coefficient, another basis monomial namely
\begin{gather}
\be_\ju \be_\ku = \pm \be_\lu \quad \mbox{where} \quad 
                  \lu = (\ju \cup \ku ) \setminus (\ju \cap \ku ).
\label{eq:monomials}
\end{gather}  
Thus, by part (ii), the inner product $\lang \be_\iu, \be_\ju\be_\ku \rang = \pm \lang \be_\iu, \be_\lu \rang$ is nonzero if and only if the index lists $\iu$ and $\lu$ sorted by $<$ are identical. This means, the product is non zero if and only if $\be_\ju \be_\ku = \pm \be_\iu$ which in turn is true if and only if $\be_\ju \be_\iu = \pm \be_\ku$. Therefore, 
$\lang \be_\iu, \be_\ju\be_\ku \rang =  \lang \tilde{\be}_\ju \be_\iu, \be_\ku \rang = 0$ if and only if $\be_\iu \neq \pm \be_\ju \be_\ku.$

To prove the identity 
\begin{equation}
\lang \be_\iu, \be_\ju\be_\ku \rang = \lang \tilde{\be}_\ju \be_\iu, \be_\ku \rang
\label{eq:theidentity}
\end{equation}
when $\be_\iu = \pm \be_\ju \be_\ku$, we proceed by induction on the grade $s=|\ju|$ of $\be_\ju.$ We establish first two base cases for $s=0$ and $s=1.$ Certainly, when $\be_\ju = 1$, then 
$\lang \be_\iu, 1\be_\ku \rang = \lang \tilde{1} \be_\iu, \be_\ku \rang 
                                   =  \lang 1 \be_\iu, \be_\ku \rang$. 

Let $\be_\ju = \be_{j_1} \in \cb{B}_1$.  Then, for any 
$\be_\ku \in \cb{B}$, we have from~(\ref{eq:xu}): 
\begin{gather}
\be_\ju \be_\ku = \be_{j_1} \be_\ku = 
\underbrace{\be_{j_1} \JJ_B \be_\ku}_{\mathrm{grade} = |\ku|-1} + 
\underbrace{\be_{j_1} \wedge \be_\ku}_{\mathrm{grade} = |\ku|+1}.     
\label{eq:j1k}
\end{gather}
Therefore, 
\begin{gather}
\lang \be_\iu, \be_\ju\be_\ku \rang =  
\lang \be_\iu, \underbrace{\be_{j_1} \JJ_B \be_\ku}_{\mathrm{grade} = |\ku|-1}\rang +
\lang \be_\iu, \underbrace{\be_{j_1} \wedge \be_\ku}_{\mathrm{grade} = |\ku|+1}\rang,
\label{eq:split}
\end{gather} 
so we consider two cases remembering that $\lang \be_\iu, \be_\ju\be_\ku \rang \neq 0$.\\

\underline{Case~1} ($|\iu| = |\ku|-1$) In this case, $\lang \be_\iu, \be_{j_1} \wedge \be_\ku \rang  = 0$ since the grades of the arguments are different. Also, $j_1 \in \ku$ assuring 
$\be_{j_1} \JJ_B \be_\ku \neq 0$, and $j_1 \notin \iu$ as then $\lang \be_\iu, \be_{j_1} \JJ_B \be_\ku \rang $ would be zero otherwise. Thus, the l.h.s. of~(\ref{eq:theidentity}) is 
$\lang \be_\iu, \be_\ju\be_\ku \rang =  \lang \be_\iu, \be_{j_1} \JJ_B \be_\ku \rang$. 

Using the symmetry of the inner product twice together with the definition of the Clifford product~(\ref{eq:xu}) and Lounesto's duality~(\ref{eq:Lduality}), we compute the r.h.s of~(\ref{eq:theidentity}) as follows:
\begin{align*}
\lang \tilde{\be}_\ju \be_\iu, \be_\ku \rang 
&= \lang \tilde{\be}_{j_1} \be_\iu, \be_\ku \rang =
\lang \be_{j_1} \be_\iu, \be_\ku \rang =
\lang \underbrace{\be_{j_1} \JJ_B \be_\iu}_{\mathrm{grade} = |\iu|-1} + 
\underbrace{\be_{j_1} \wedge \be_\iu}_{\mathrm{grade} = |\iu|+1}, \be_\ku \rang\\
&= \lang \be_{j_1} \wedge \be_\iu, \be_\ku  \rang =
   \lang \be_\ku, \be_{j_1} \wedge \be_\iu \rang  = 
   \lang \be_{j_1} \JJ_B \be_\ku, \be_\iu \rang\\
&= \lang \be_\iu, \be_{j_1} \JJ_B \be_\ku \rang =
   \lang \be_\iu, \be_{j_1} \JJ_B \be_\ku + \be_{j_1} \wedge \be_\ku \rang =
   \lang \be_\iu, \be_\ju \be_\ku \rang 
\end{align*}
since $|\ku|=|\iu|+1$ and $j_1 \notin \iu$ so $\be_{j_1} \JJ_B \be_\iu =0$.\\

\underline{Case~2} ($|\iu| = |\ku|+1$) In this case, $\lang \be_\iu, \be_{j_1} \JJ_B \be_\ku \rang  = 0$ since the grades of the arguments are different. Also, $j_1 \notin \ku$ as then 
$\lang \be_\iu, \be_{j_1} \wedge \be_\ku \rang$ would be zero otherwise, and $j_1 \in \iu$ assuring that $\be_{j_1} \wedge \be_\iu = \tilde{\be}_{j_1} \wedge \be_\iu = 0.$ Thus, the l.h.s. 
of~(\ref{eq:theidentity}) equals:
\begin{align*}
\lang \be_\iu, \be_\ju\be_\ku \rang 
&= 
\lang \be_\iu, \underbrace{\be_{j_1} \wedge \be_\ku}_{\mathrm{grade} = |\ku|+1} \rang =
\lang \be_\iu, \underbrace{\tilde{\be}_{j_1} \wedge \be_\ku}_{\mathrm{grade} = |\ku|+1}\rang
= \lang \tilde{\be}_{j_1} \JJ_B \be_\iu, \be_\ku \rang\\ 
&= 
\lang \tilde{\be}_{j_1} \JJ_B \be_\iu + \tilde{\be}_{j_1} \wedge \be_\iu, \be_\ku \rang =
\lang \tilde{\be}_{j_1} \be_\iu, \be_\ku \rang  =
\lang \tilde{\be}_\ju \be_\iu, \be_\ku \rang  
\end{align*} 
which is the r.h.s. of~(\ref{eq:theidentity}).

Having established the two base cases, let us now assume that the identity~(\ref{eq:theidentity}) is true for any $\be_\ju$ such that $0 \leq |\ju| \leq s-1$ and for any $\be_\iu,\be_\ku \in \cb{B}$ as long as $\be_\iu = \pm \be_\ju\be_\ku$. We will now show that it is also true when $|\ju|=s$.

Let $\be_\ju = \be_{j_1} \be_{j_2} \cdots \be_{j_{s-1}} \be_{j_s}$. Then, 
\begin{align*}
\hspace*{4ex}\lang \be_\iu,\be_\ju \be_\ku \rang 
&=%
\lang \be_\iu,\be_{j_1} \be_{j_2} \cdots \be_{j_{s-1}} \be_{j_s} \be_\ku \rang =
\lang \be_\iu,(\be_{j_1} \be_{j_2} \cdots \be_{j_{s-1}})(\be_{j_s} \be_\ku) \rang \\
&=%
\lang (\be_{j_1} \be_{j_2} \cdots \be_{j_{s-1}})\tilde{} \;\be_\iu,\be_{j_s} \be_\ku  \rang = 
\lang (\be_{j_s})\tilde{}\; (\be_{j_1} \be_{j_2} \cdots \be_{j_{s-1}})\tilde{} \;\be_\iu,
\be_\ku  \rang\\
&=%
\lang (\be_{j_1} \be_{j_2} \cdots \be_{j_{s-1}} \be_{j_s})\tilde{}\; \be_\iu,\be_\ku \rang =
\lang \tilde{\be}_\ju \be_\iu, \be_\ku \rang. 
\end{align*}
Thus, we have proven identity~(\ref{eq:theidentity}) for any basis monomials
$\be_\iu,\be_\ju, \be_\ku \in \cb{B}.$

(iv) This result follows easily from bilinearity of the inner product, properties of the Clifford product, part (iii) of this Lemma, and the linearity of the reversion anti-involution.

Let $u = \sum_{\iu \in 2^{[n]}} u_\iu \be_\iu$, $v = \sum_{\ju \in 2^{[n]}} v_\ju \be_\ju$, and $w = \sum_{\ku \in 2^{[n]}} w_\ku \be_\ku$ be any basis monomials in~$\cb{B}$. Then,
$$
\lang u, vw \rang = 
\sum_{\iu,\ju,\ku} u_\iu v_\ju w_\ku \lang \be_\iu, \be_\ju \be_\ku \rang =
\sum_{\iu,\ju,\ku} u_\iu v_\ju w_\ku \lang \tilde{\be}_\ju \be_\iu, \be_\ku \rang =
\lang \tilde{v}u,w \rang.  
$$ 
\renewcommand{\qedsymbol}{}
\end{proof}


\end{document}

%% file: Tmacros.tex
\newcommand{\cl}{C \kern -0.1em \ell}     

\newcommand{\JJ}{\mathbin{\raisebox{0.25ex}{$\footnotesize
                       \rm\vphantom{I}%
                       \_\hskip -0.25em\_%
                       \vrule width 0.6pt$}}}           

\newcommand{\w}{\wedge}

\newcommand{\Mat}{{\rm Mat}}

  \newcommand{\be}{{\bf e}}

\newcommand{\bL}{{\bf L}}

  \newcommand{\bu}{{\bf u}}
  \newcommand{\bv}{{\bf v}}
  
  \newcommand{\bx}{{\bf x}}
  \newcommand{\by}{{\bf y}}

\newcommand{\BK}{\mathbb{K}}

\newcommand{\BR}{\mathbb{R}}

\newcommand{\BZ}{\mathbb{Z}}
\newcommand{\diag}{\mbox{\rm diag}}
\newcommand{\spn}{\mbox{\rm span}}

\newcommand{\beq}{\begin{equation}}
\newcommand{\eeq}{\end{equation}}

